\documentclass[11pt]{scrartcl}
\usepackage[disable]{todonotes}

\usepackage[utf8]{inputenc} 
\usepackage{mathtools, amsfonts, amssymb, breqn}
\usepackage{amsthm}

\usepackage[british]{babel}
\usepackage{graphicx, hyperref, tabularx, abstract, makecell,ragged2e}

\usepackage[format=plain]{caption}
\DeclareCaptionFormat{cont}{#1 (cont.)#2#3\par}

\usepackage[outdir=./]{epstopdf}   

\usepackage{cellspace, hhline} 

\usepackage{subfig}
\usepackage{xcolor-solarized}
\usepackage{todonotes}
\usepackage{overpic}

\usepackage{enumitem}
\setdescription{itemsep=5pt,parsep=0pt,leftmargin=1.55cm,font=\normalfont}

\usepackage[]{geometry}
\newcommand{\N}{\ensuremath{\mathbb{N}}}
\newcommand{\R}{\ensuremath{\mathbb{R}}}
\newcommand{\C}{\ensuremath{\mathbb{C}}}

\newcommand{\mattwo}[4]{\left(
	\begin{array}{cc}
		#1 & #2 \\
		#3 & #4 \\  
	\end{array}
	\right)
}

\newcommand{\dif}{{\operatorname{d}}}
\newcommand{\tr}{{\operatorname{tr}}}
\newtheorem{prop}{Proposition}[section]

\newtheorem{cor}[prop]{Corollary}

\usepackage{cellspace, hhline}
\definecolor{changes}{RGB}{0, 0, 0}
\numberwithin{equation}{section}
\numberwithin{figure}{section}
\numberwithin{table}{section}

\author{L. Eigentler and J. A. Sherratt}

\newcommand{\Addresses}{{% additional braces for segregating \footnotesize
		\bigskip
		\footnotesize
		
		L. Eigentler (corresponding author),  Maxwell Institute for Mathematical Sciences, Department of Mathematics, Heriot-Watt University, Edinburgh EH14 4AS, United Kingdom\par\nopagebreak
		\textit{E-mail address}: \texttt{le8@hw.ac.uk}
		
		\medskip
		
	J.A. Sherratt,  Maxwell Institute for Mathematical Sciences, Department of Mathematics, Heriot-Watt University, Edinburgh EH14 4AS, United Kingdom\par\nopagebreak
	\textit{E-mail address}: \texttt{J.A.Sherratt@hw.ac.uk}

}}

\title{Effects of precipitation intermittency on vegetation patterns in semi-arid landscapes}
\date{\vspace{-5ex}}
\begin{document}

%%%%%%%%%%%%%%%%%%%%%%%%%%%%%%%%%%%%%%%%%%%%%%%%%%%%%%%%%%%%%%%%%%%%%%%%%%%%%%%%%%%%%%%%%%%%%%%%%%%%%%%%%%%%%%%%%%%%%%%%%%%
%%%%%%%%%%%%%%%%%%%%%%%%%%%%%%             INTRODUCTION             %%%%%%%%%%%%%%%%%%%%%%%%%%%%%%%%%%%%%%%%%%%%%%%%%%%%%%%%
%%%%%%%%%%%%%%%%%%%%%%%%%%%%%%%%%%%%%%%%%%%%%%%%%%%%%%%%%%%%%%%%%%%%%%%%%%%%%%%%%%%%%%%%%%%%%%%%%%%%%%%%%%%%%%%%%%%%%%%%%%% 

\maketitle
\Addresses
\listoftodos
%\maketitle
\begin{abstract} 
  	Patterns of vegetation are a characteristic feature of many semi-arid regions. The limiting resource in these ecosystems is water, which is added to the system through short and intense rainfall events that cause a pulse of biological processes such as plant growth and seed dispersal. We propose an impulsive model based on the Klausmeier reaction-advection-diffusion system, analytically investigate the effects of rainfall intermittency on the onset of patterns, and augment our results by numerical simulations of model extensions. Our investigation focuses on the parameter region in which a transition between uniform and patterned vegetation occurs. Results show that decay-type processes associated with a low frequency of precipitation pulses inhibit the onset of patterns and that under intermittent rainfall regimes, a spatially uniform solution is sustained at lower total precipitation volumes than under continuous rainfall, if plant species are unable to efficiently use low soil moisture levels. Unlike in the classical setting of a reaction-diffusion model, patterns are not caused by a diffusion-driven instability but by a combination of sufficiently long periods of droughts between precipitation pulses and water diffusion. Our results further indicate that the introduction of pulse-type seed dispersal weakens the effects of changes to width and shape of the plant dispersal kernel on the onset of patterns. 
\end{abstract}

MSC codes: 35R09, 35R12, 35B36 

Keywords: impulsive model, nonlocal dispersal, pattern formation, semi-arid landscapes, precipitation intermittency

\section{Introduction}

Self-organised vegetation patterns are a characteristic feature of semi-arid regions around the world. The formation of patterns is caused by a positive feedback between plant growth and water redistribution towards areas of high biomass \cite{Rietkerk2008}. Mechanisms involved in the establishment of such a feedback loop include the formation of infiltration-inhibiting soil crusts in areas of bare ground that induce overland water flow, a combination of strong local water uptake (vertically extended root systems) and fast soil water diffusion, nonlocal water uptake (laterally extended root systems), or a combination thereof \cite{Meron2016}. Redistribution of water towards dense biomass patches drives further plant growth in these regions and thus closes the feedback loop. First discovered through areal photography in the 1950s \cite{Macfadyen1950}, vegetation patterns have been detected in various semi-arid regions of the world (see \cite{Valentin1999, Gandhi2019} for reviews) such as in the African Sahel \cite{Mueller2013, Deblauwe2012}, Somalia \cite{Hemming1965, Gowda2018}, Australia \cite{Dunkerley2002, Heras2012}, Israel \cite{Sheffer2013, Buis2009}, Mexico and the US \cite{Deblauwe2012, Penny2013, Pelletier2012} and northern Chile \cite{Fernandez-Oto2019a}. The understanding of the evolution and underlying dynamics of patterned vegetation is of crucial importance as changes to properties such as pattern wavelength, recovery time from perturbations or the area fraction colonised by plants may provide an early indication of an irreversible transition to full desert \cite{Kefi2007, Rietkerk2004, Corrado2014, Gowda2016, Meron2018, Dakos2011, Saco2018, Zelnik2018}. 

The amount of empirical data on vegetation patterns is limited due to the inability to reproduce patterns in a laboratory setting and the long time scale involved in the formation and evolution of them. Thus, a range of different mathematical models describing the phenomenon have been proposed {\color{changes}(in particular by Rietkerk et al. \cite{Rietkerk2002} and Gilad et al. \cite{Gilad2004}), which focus} on various different processes that are involved in the formation of vegetation patterns. One model that stands out due to its deliberately basic description of the plant-water-dynamics in semi-arid environments is the Klausmeier model \cite{Klausmeier1999}. The excellent framework for mathematical analysis and model extensions provided by the Klausmeier reaction-advection-diffusion model has been utilised extensively in the past (e.g. \cite{Bennett2019,Sherratt2005,Sherratt2007,Sherratt2010,Sherratt2011,Sherratt2013III,Sherratt2013IV,Sherratt2013V,Siteur2014,Ursino2006, Eigentler2018nonlocalKlausmeier, Eigentler2019Multispecies,Wang2019,Wang2018a, Consolo2019, Consolo2019a, Marasco2014, Bastiaansen2018, Eigentler2019integrodifference, Eigentler2020savanna_coexistence}). 

Rainfall in semi-arid regions occurs intermittently, seasonally or as a combination of both. Under intermittent rainfall regimes only a small number of short-lasting precipitation events per year provide a sufficiently large amount of water to affect vegetation growing in these regions \cite{Noy-Meir1973}. {\color{changes} If such rainfall events are sufficiently separated, they cause a pulse of biological processes before decay-type phenomena of dry spells take over \cite{Noy-Meir1973}. Besides plant growth, seed dispersal is also commonly observed to be synchronised with precipitation events. One mechanism, widespread in dryland ecosystems, which causes such a behaviour is ombrohydrochory, the opening of a seed container due to contact with water \cite{Navarro2009, RheedevanOudtshoorn2013}. }Plants in semi-arid regions are sensitive to quantity, frequency and temporal spread of intermittent precipitation events \cite{Fravolini2005, Fan2016, Liu2012, Heisler-White2008}.

Experimental studies suggest that if the total precipitation volume is kept constant, then a lower frequency of rainfall events yields higher plant biomass \cite{Lundholm2004}, an increase in the aboveground net primary productivity \cite{Heisler-White2008} and an increase in the seedlings' survival rate \cite{Sher2004}. The main factor for these beneficial effects is the temporal increase in soil moisture caused by larger rain events, while a higher number of smaller individual events keeps the moisture level below a threshold needed for the activation of biological processes in plants \cite{Heisler-White2008, Sher2004}. Contradictory evidence regarding seedling survival exists, which suggests that the effects of rain intermittency depend on a range of factors \cite{Lundholm2004}. In the future, changes to the temporal variability of precipitation (in particular the intensity of rainfall events) are expected to occur globally \cite{IPCC2014,Easterling2000}. 

Despite the fact that seasonality, intermittency and intensity of precipitation has an important influence on semi-desert ecosystems, most mathematical models assume that rainfall occurs continuously {\color{changes} and uniformly} in time. Some simulation-based studies, however, have addressed the phenomenon by introducing seasonality (i.e. a wet and dry season) and intermittency of rainfall to existing models for dryland vegetation dynamics. {\color{changes}These include modifications of the Rietkerk model \cite{Guttal2007, Siteur2014a}, of the Klausmeier model \cite{Ursino2006} as well as of the Gilad model \cite{Kletter2009}.} Baudena et al. \cite{Baudena2007,Baudena2008} couple a model describing the soil moisture proposed by Laio \cite{Laio2001} for the upper soil layer, in which water is added during a wet season either at a constant rate or as an instantaneous event, to vegetation dynamics. The results of these studies show beneficial effects of rainfall intermittency, such as an increase in the area covered by vegetation \cite{Kletter2009} or plant biomass \cite{Kletter2009, Baudena2007} but also suggest that a lower frequency of rain pulses increases the minimum requirement on the total annual precipitation needed to avoid convergence to a bare soil state \cite{Ursino2006}. The latter result also suggests that the size of the parameter region in which pattern onset occurs reduces under intermittent rainfall regimes \cite{Ursino2006}.
Seasonality of precipitation may have similar effects \cite{Guttal2007}, but can also be detrimental to plants by reducing their biomass and area fraction covered \cite{Kletter2009, Baudena2007}. 

{\color{changes}Effects due to changes in the frequency of rainfall events have received very little attention in the mathematical modelling of vegetation patterns, with the works of Ursino and Contarini \cite{Ursino2006} on the Klausmeier model, Kletter et al.\cite{Kletter2009} on the Gilad model, and Siteur et al. \cite{Siteur2014a} on the Rietkerk model being notable exceptions. However, none of these papers consider both a wide range of biologically relevant interpulse times and the system dynamics in drought periods. For example, both Kletter et al. and Ursino and Contarini restrict their investigation to a small number of different precipitation frequencies, while Siteur et al. neglect the ecohydrological dynamics between rainfall events. Moreover, most theoretical approaches to study temporal non-uniformity in precipitation are simulation-based. In this paper, we introduce a model based on the Klausmeier model that captures both the impulsive nature of precipitation pulses and associated processes, and also the drought period dynamics. We keep our model sufficiently simple to allow for an analytical investigation of pattern onset in the system. This enables us to consider a wide range of different rainfall regimes and study the effects of precipitation intermittency on the ecohydrological dynamics.}

One approach to modelling a system in which pulse-type phenomena occur is the use of integrodifference equations. {\color{changes} In separate work \cite{Eigentler2019integrodifference}, we show that such a framework is insufficient to capture effects of precipitation intermittency as it is unable to account for the dynamics specific to drought periods between rainfall pulses. To instead describe situations in which pulse-type phenomena occur alongside the continuous processes of dry spells,} impulsive-type systems are used. Such models consist of a system of PDEs describing continuous processes on a finite time domain $(n-1)T<t<nT$, $n\in\N$ and a set of discrete equations that update the densities at times $nT$. The use of impulsive models is a relatively new approach in mathematical modelling but such models are suitable for the description of a wide range of systems. Previous applications include descriptions of populations whose life cycle consists of two non-overlapping stages, such as organisms whose larvae are subjected to a water flow \cite{Vasilyeva2016, Huang2017}; predator prey systems in which consumer reproduction occurs only once a year and is based on the amount of stored energy accumulated through consumption of prey during the year \cite{Wang2018} or that are periodically subjected to external inputs \cite{Akhmet2006}; and more general consumer-resource systems in which the consumer reproduction is synchronised \cite{Pachepsky2008, Lewis2012} or in which seasonal harvesting occurs \cite{Lewis2012}. Impulsive models can further provide a mechanistic interpretation of the underlying ecological processes involved in purely discrete systems \cite{Geritz2004}.

The modelling of plant dispersal as an instantaneous event requires its description by a convolution integral instead of the widely used and mathematically more accessible diffusion term. Biologically, however, this provides a more realistic description of the spatial redistribution of plants as the dynamics of seed dispersal are often affected by nonlocal processes \cite{Bullock2017}. The use of a convolution term in the description of seed dispersal is thus not a novelty of this paper but has been also been used in a number of previous models for dryland ecosystems \cite{Baudena2013, Pueyo2008, Pueyo2010, Eigentler2018nonlocalKlausmeier}.

%%%%%%%%%%%%%%%%%%%%%%%%%%%%%%%%%%%%%
In this paper, we introduce and analytically study an impulsive model based on the Klausmeier model to gain a better understanding of the effects of pulse-type processes on the onset of vegetation patterns. We motivate the presentation of the model in Section \ref{sec: models} by a review of the Klausmeier model and its most relevant results.
In Section \ref{sec: Impulsive} we derive conditions for the onset of patterns in the impulsive model based on a linear stability analysis. This allows us to investigate how changes in the rainfall regime affect pattern onset and provides an insight into the mechanism that is responsible for the formation of patterns in the model.
The analysis presented in Section \ref{sec: Impulsive} is tractable due to some simplifications, such as the use of a specific plant dispersal kernel and the restriction to a flat domain. In Section \ref{sec: Impulsive: Simulations} we augment our analytical results by numerical simulations of extensions of the basic model studied in Section \ref{sec: Impulsive} to analyse and discuss the effects of our simplifying assumptions. 
We present an interpretation of our results and address potential shortfalls in Section \ref{sec: Discussion}.

%%%%%%%%%%%%%%%%%%%%%%%%%%%%%%%%%%%%%%%%%%%%

\section{Model description}\label{sec: models}
In this section we introduce the model which we use to investigate the effects of rainfall intermittency on the onset of patterns in semi-arid environments. We base our model on an extension of the Klausmeier model, whose most relevant results are reviewed.
\subsection{Klausmeier models}
One of the earliest models describing the plant-water dynamics in semi-arid environments is due to Klausmeier \cite{Klausmeier1999}. The relative simplicity of the model provides a framework for a rich mathematical analysis (e.g. \cite{Sherratt2005,Sherratt2007,Sherratt2010,Sherratt2011,Sherratt2013III,Sherratt2013IV,Sherratt2013V,Siteur2014,Ursino2006}). After a suitable nondimensionalisation \cite{Klausmeier1999,Sherratt2005} the model is
\begin{subequations}
	\label{eq: Intro Klausmeier local}
	\begin{align}
	\frac{\partial u}{\partial t} &= \overbrace{u^2w}^{\text{plant growth}} - \overbrace{Bu}^{\text{plant mortality}} + \overbrace{\frac{\partial^2 u}{\partial x^2}}^{\text{plant dispersal}}, \label{eq:Intro Klausmeier local plants}   \\
	\frac{\partial w}{\partial t} &= \underbrace{A}_{\text{rainfall}}-\underbrace{w}_{\text{evaporation}}-\underbrace{u^2w}_{\substack{\text{water consumption}\\\text{by plants}}} + \underbrace{\nu\frac{\partial w}{\partial x}}_{\substack{\text{water flow}\\\text{downhill}}}+ \underbrace{d\frac{\partial^2w}{\partial x^2}}_{\text{water diffusion}},
	\end{align}
\end{subequations}
where $u(x,t)$ denotes the plant density, $w(x,t)$ the water density, $x\in \R$ the space domain where $x$ is increasing in the uphill direction and $t>0$ denotes the time. The diffusion of water was not originally included in the model but is a well established  addition \cite{Kealy2012,Siteur2014, Stelt2013, Zelnik2013}. It is assumed that water is added to the system at a constant rate, evaporation effects are proportional to the water density \cite{Rodriguez-Iturbe1999, Salvucci2001} and the plant mortality rate is density-independent. {\color{changes}The nonlinearity in the water consumption and plant growth terms arises due to the positive feedback between local vegetation growth and water redistribution. Water uptake by plants is the product of the consumer density ($u$), the resource density ($w$) and a term that accounts for the increased resource availability due to the positive feedback caused, for example, by an increase of soil permeability in vegetated areas ($u$). This nonlinearity} drives the formation of spatial patterns. The parameters $A$, $B$, $\nu$ and $d$ represent rainfall, plant loss, the slope and water diffusion, respectively. 

{\color{changes}The Klausmeier model \eqref{eq: Intro Klausmeier local} combines all hydrological dynamics into one single variable $w$. By contrast, some other modelling frameworks distinguish between surface water and soil moisture dynamics \cite{Rietkerk2002, Gilad2004}. In this paper, we focus on the modelling framework presented by the Klausmeier model without such a distinction, but the application of our modelling approach to a system with both surface and soil water density is briefly discussed in Sec. \ref{sec: Discussion}.}

In a previous paper \cite{Eigentler2018nonlocalKlausmeier}, we have studied the effects of replacing the plant diffusion term in the Klausmeier model by a convolution of a probability density $\phi$ and the plant density $u$, i.e. 
\begin{subequations}
	\label{eq: Intro Klausmeier nonlocal}
	\begin{align}
	\frac{\partial u}{\partial t} &= u^2w - Bu + C\left( \phi(\cdot;a) \ast u(\cdot,t) - u(x,t)\right), \label{eq:Intro Klausmeier nonlocal plants}   \\
	\frac{\partial w}{\partial t} &= A-w-u^2w + \nu\frac{\partial w}{\partial x}+ d\frac{\partial^2w}{\partial x^2}.
	\end{align}
\end{subequations}
The additional parameters $C$ and $a$ represent the rate of plant dispersal and reciprocal width of the dispersal kernel, respectively. 

A linear stability analysis of both the local model \eqref{eq: Intro Klausmeier local} and the nonlocal model \eqref{eq: Intro Klausmeier nonlocal}, with the Laplace kernel
\begin{align}\label{eq:difference equations: simulations: Laplacian}
\phi(x) = \frac{a}{2}e^{-a|x|}, \quad a>0, x\in\R,
\end{align}
used in the latter, gives an insight into the nature of patterned solutions of the system. On flat ground, i.e. $\nu = 0$, the onset of spatial patterns occurs due to a diffusion-driven instability. Thus for any level of rainfall $A$, there exists a threshold $d_c \in \R$ on the diffusion coefficient such that an instability occurs for all $d>d_c$. The analysis for the nonlocal model with the Laplace kernel shows that an increase in the width of the dispersal kernel inhibits the formation of patterns by causing an increase in the diffusion threshold.

Unlike the Laplace kernel, other kernel functions do not provide a simplification sufficient to study the onset of patterns analytically. Numerical simulations, however, confirm that the trends observed from the linear stability analysis for the Laplace kernel also apply to other kernel functions \cite{Eigentler2018nonlocalKlausmeier}.

\subsection{Impulsive Model} 

The Klausmeier model assumes that all processes occur continuously in time. To account for the more realistic combination of pulse-type events associated with short, high intensity precipitation events with the continuous nature of plant loss, water evaporation and water dispersal, we propose an impulsive model to describe the plant and water dynamics in semi-arid environments.
Under the assumption that water transport and the decay-type processes of plant mortality and water evaporation are the only processes occurring in drought periods between rainfall pulses \cite{Noy-Meir1973}, the model is
%In order to study impulsive models analytically, it is desirable that the continuous time, ``during the year'' equations are solvable at least in the case of no space dependence. If this is the case, then the model can be reduced to a discrete time difference model in terms of the densities at the beginning of each ``year''. One of the simplest cases in which this is possible is if the PDEs are linear and not coupled. In its most general form the model then is \todo{Motivate the simplicity of the model in a different way.}
\begin{subequations}
	\label{eq: Impulsive: intro: model general nondim}
	\begin{align}
	\frac{\partial u_n}{\partial t} &= \overbrace{-k_1u_n}^{\text{plant loss}}, \label{eq: Impulsive: intro: model general nondim u in year}\\
	\frac{\partial w_n}{\partial t} &= \underbrace{-k_2w_n}_{\text{evaporation}} + \underbrace{k_3 \frac{\partial w_n}{\partial x}}_{\substack{\text{water flow}\\\text{downhill}}} + \underbrace{k_4 \frac{\partial^2 w_n}{\partial x^2}}_{\text{water diffusion}}, \label{eq: Impulsive: intro: model general nondim  w in year}\\
	u_{n+1}(x,0) &= \tilde{f}\left(u_n(x,\tau),w_n(x,\tau)\right), \label{eq: Impulsive: intro: model general nondim u between year}\\
	w_{n+1}(x,0) &= \tilde{g}\left(u_n(x,\tau),w_n(x,\tau)\right), \label{eq: Impulsive: intro: model general nondim w between year}
	\end{align}
\end{subequations}
where $u_n=u_n(x,t)$, $w_n=w_n(x,t)$, $x \in \R$, $0<t<\tau$ and $n \in \N$. The spatial domain is considered to be infinite with $x$ increasing in the uphill direction. Between the $(n-1)$-th and $n$-th precipitation pulse, the interpulse PDEs \eqref{eq: Impulsive: intro: model general nondim u in year} and \eqref{eq: Impulsive: intro: model general nondim  w in year} are considered on the finite time domain $0<t<\tau$, where $\tau$ is the time (in years) between the occurrence of the pulse events described by the update equations \eqref{eq: Impulsive: intro: model general nondim u between year} and \eqref{eq: Impulsive: intro: model general nondim w between year}. The interpulse PDEs \eqref{eq: Impulsive: intro: model general nondim u in year} and \eqref{eq: Impulsive: intro: model general nondim  w in year} describe the continuous loss of plants at rate $k_1$, and evaporation at rate $k_2$. While no plant dispersal is assumed to occur during this phase, water diffuses with diffusion coefficient $k_4$ and flows downhill at velocity $k_3$. The simplistic nature of the PDE system allows for an analytical study of conditions for pattern onset to occur (Section \ref{sec: Impuslsive: LinStab}), but an extension which also includes plant growth during drought periods is considered using numerical simulations in Section \ref{sec: Impulsive: Simulations}.

The functions $\tilde{f}(u_n(x,\tau),w_n(x,\tau))$ and $\tilde{g}(u_n(x,\tau),w_n(x,\tau))$ in the update equations \eqref{eq: Impulsive: intro: model general nondim u between year} and \eqref{eq: Impulsive: intro: model general nondim w between year} describe the system's dynamics during short rainfall pulses, which are assumed to occur periodically in time. To account for plant growth and the associated consumption of water as well as seed dispersal synchronised with a precipitation event, we choose
\begin{align*}
\tilde{f}\left(u_n(x,\tau),w_n(x,\tau)\right) &= \overbrace{u_n(x,\tau)}^{\text{existing plants}} + \overbrace{\phi(\cdot;a) \ast \left(k_5\left(\frac{u_n(\cdot,\tau)}{k_6+u_n(\cdot,\tau)}\right)^2 \left(w_n(\cdot,\tau) + \tau k_7\right) \right)}^{\text{dispersal of newly added biomass}},\\
\tilde{g}\left(u_n(x,\tau),w_n(x,\tau)\right) &= \underbrace{w_n(x,\tau)}_{\text{existing water}} + \underbrace{\tau k_7}_{\text{rainfall}} -  \underbrace{\left(\frac{u_n(x,\tau)}{k_6+u_n(x,\tau)}\right)^2\left(w_n({\color{changes}x},\tau) + \tau k_7\right)}_{\text{water uptake}}.
\end{align*}
  In the update equation \eqref{eq: Impulsive: intro: model general nondim w between year} a constant amount water $\tau k_7$ is added to the existing water density. The parameter $k_7$ denotes the total amount of rainfall that occurs over one year and $\tau$ (in years) is the time between two rainfall events. The water volume added to the system during one precipitation event thus is $\tau k_7$. At the same time, water is converted into biomass. 
  
  {\color{changes} Similar to the Klausmeier model \eqref{eq: Intro Klausmeier local}, the term describing water consumption by plants consists of the total resource density ($w+\tau k_7$), a term describing the water uptake by the plants' roots ($u/(k_6+u)$), and a term accounting for the increased ability of plants to consume water in dense patches $u/(k_6+u)$). As in the Klausmeier model, the functional responses of the latter two to the plant density are chosen to be identical for mathematical convenience. However, the functional response is different to that used in the Klausmeier model. In the impulsive model $H_{\operatorname{up}}(u) = u/(k_6+u)$, motivated by the saturating behaviour of water infiltration into the soil based on empirical evidence \cite{Wijngaarden1985} and previous applications in mathematical models \cite{HilleRisLambers2001, Gilad2004}, while in the Klausmeier model $H_{\operatorname{up}}(u) = u$. The pulse-type occurrence of precipitation and water uptake in \eqref{eq: Impulsive: intro: model discrete dispersal rain water uptake} necessitates a saturating behaviour ($H_{\operatorname{up}}^2(u) <1$ for all $u\ge0$) of the functional response to ensure positivity of \eqref{eq: Impulsive: intro: model discrete dispersal rain water uptake w between year}. The parameter $k_6$ is the half saturation constant of the water infiltration and corresponds to the level of plant biomass at which the water infiltration into the soil is at half of its maximum. This water uptake term directly corresponds to the term in \eqref{eq: Impulsive: intro: model general nondim u between year}, describing plant growth, where $k_5$ quantifies the plant species' water to biomass conversion rate.  We have numerically tested the model for other nonlinearities in this term with such a saturating behaviour without observing any qualitative differences in the results on pattern onset.
  
  Finally, dispersal of the newly added biomass is described by the convolution term of that biomass with a probability density function $\phi$. This introduces an additional parameter $a$, describing the width of the dispersal kernel in a reciprocal way. This constitutes a second main difference to the models discussed above. While in the Klausmeier models the whole plant density undergoes diffusion/nonlocal dispersal, in the impulsive model only newly added biomass is dispersed, weakening the role of dispersal in the model. }
  
  No water redistribution is assumed to occur in this stage. While overland water flow during intense rainfall events is an area of active research \cite{Rossi2017, Thompson2011, Wang2015}, some hydrological modelling approaches suggest that if the contrast in water infiltration rates between bare and vegetated soil is small (e.g. in non crust-forming soil types such as sandy soil), then no water run-on occurs at plant patches during precipitation pulses \cite{Rossi2017}. An overview of all parameters, including estimates, is given in Table \ref{tab: Implsive: Intro: Parameters}. 

\begin{table}
	\begin{tabularx}{\textwidth}{lllX}
		\multicolumn{4}{l}{\textbf{Dimensional parameters of \eqref{eq: Impulsive: intro: model general nondim}}} \\
		\hline
		Parameter & Units & Estimates & Description \\
		\hline
		$k_1$ & year$^{-1}$ 						& \makecell[l]{1.8 \cite{Klausmeier1999}, 0.18 \cite{Klausmeier1999}, \\ 1.2\cite{Gilad2007}}	& Rate of plant loss  \\ \hline
		$k_2$ & year$^{-1}$ 						& \makecell[l]{4 \cite{Klausmeier1999,Gilad2007,Siteur2014}, \\ 0.2\cite{Rietkerk2002} }						& Rate of evaporation \\ \hline
		$k_3$ & m year$^{-1}$ 						& 0-365 \cite{Klausmeier1999} & Velocity of water flow downhill \\ \hline
		$k_4$ & m$^2$ year$^{-1}$ 					& 500 \cite{Siteur2014}, & Water diffusion coefficient \\ \hline
		$k_5$ & \makecell[l]{(kg biomass) \\ (kg H$_2${\color{changes}O})$^{-1}$} 	& \makecell[l]{0.01\cite{Rietkerk2002}, 0.003 \cite{Klausmeier1999}, \\0.002 \cite{Klausmeier1999}} & Yield of plant biomass per kg water \\ \hline
		$k_6$ & \makecell[l]{(kg biomass) \\ m$^{-2}$} 				& 0.05 \cite{Gilad2007} & Half saturation constant of water uptake  \\ \hline
		$k_7$ & \makecell[l]{(kg H$_2${\color{changes}O}) \\ m$^{-2}$ year$^{-1}$}	& \makecell[l]{250-750 \cite{Klausmeier1999},\\ 0-1000 \cite{Gilad2007}} & Total amount of rainfall in one year\\ \hline
		$\tau$ & year								& 0-1 & Interpulse time \\ \hline
		$a$ & m$^{-1}$						& 0.03-$100$ \cite{Bullock2017} & Scale parameter of dispersal kernel, reciprocal of the width \\ \hline \hline
		\multicolumn{4}{l}{\textbf{Nondimensional parameters of \eqref{eq: Impulsive: intro: model discrete dispersal rain water uptake}}} \\
		\hline
		Parameter & Scaling & Estimates & Description \\
		\hline
		$A$ & $k_2^{-1}k_5k_6^{-1}k_7$ 									& 0-15 \cite{Klausmeier1999,Gilad2007,Siteur2014} & Precipitation per year \\ \hline
		$B$ & $k_1k_2^{-1}$									 	& \makecell[l]{0.45 \cite{Klausmeier1999}, 0.3\cite{Gilad2007}  \\ 0.045 \cite{Klausmeier1999}} & Plant mortality rate\\ \hline
		$T$ & $k_2\tau$ 									& 0-4 & Interpulse time \\ \hline
		$\nu$ & $ k_2^{-1}k_3a$ 									& 0-$10^3$ \cite{Klausmeier1999, Bullock2017} & Slope (water flow downhill)\\ \hline
		$d$ & $k_2^{-1}k_4a^2$									& 0.1-$10^{6}$ \cite{Siteur2014, Bullock2017} &  Water diffusion coefficient
	\end{tabularx}
	\caption{Overview of parameters in \eqref{eq: Impulsive: intro: model general nondim} and \eqref{eq: Impulsive: intro: model discrete dispersal rain water uptake}. This table gives an overview of both the dimensional parameters of model \eqref{eq: Impulsive: intro: model general nondim} and the nondimensional parameters of \eqref{eq: Impulsive: intro: model discrete dispersal rain water uptake}, including their units (dimensional parameters) or scalings (nondimensional parameters), and their estimated values as well as an interpretation/description. {\color{changes}Note that parameter $k_5$ is dimensionless. However, for ease of interpretation, we distinguish between (kg biomass) and (kg H$_2$O).} The wide ranges for the water dispersal rates $\nu$ and $d$ arise from their dependence on the variations in the width $a$ of the plant dispersal kernel.}\label{tab: Implsive: Intro: Parameters}
\end{table}

{\color{changes}The formulation of \eqref{eq: Impulsive: intro: model general nondim} is based on a number of simplifying assumptions (e.g. flat terrain, no plant growth during drought periods, linear functional response to the water density in the water consumption term) to make to model analytically tractable (Sec \ref{sec: Impuslsive: LinStab}). In Sec. \ref{sec: Impulsive: Simulations}, we relax these assumptions and analyse their effects using numerical methods.}

The model can be nondimensionalised by $u = k_6\tilde{u}$, $w = k_5^{-1}k_6\tilde{w}$, $x = a^{-1}\tilde{x}$, $t=  k_2^{-1}\tilde{t}$, $A = k_2^{-1}k_5k_6^{-1}k_7$, $B = k_1k_2^{-1}$, $T = k_2\tau$, $\nu = k_2^{-1}k_3a$ and $d = k_2^{-1}k_4a^2$, to give 
\begin{subequations}
	\label{eq: Impulsive: intro: model discrete dispersal rain water uptake}
	\begin{align}
	\frac{\partial u_n}{\partial t} &= -Bu_n, \label{eq: Impulsive: intro: model discrete dispersal rain water uptake u in year}\\
	\frac{\partial w_n}{\partial t} &= -w_n + \nu \frac{\partial w_n}{\partial x} + d \frac{\partial^2 w_n}{\partial x^2}, \label{eq: Impulsive: intro: model discrete dispersal rain water uptake w in year}\\
	u_{n+1}(x,0) &= u_n(x,T) +  \phi(\cdot;1) \ast \left(\left(\frac{u_n(\cdot,T)}{1+u_n(\cdot,T)}\right)^2 \left(w_n(\cdot,T) + TA\right) \right), \label{eq: Impulsive: intro: model discrete dispersal rain water uptake u between year}\\
	w_{n+1}(x,0) &= \left(w_n(x,T) + TA\right) \left(1-\left(\frac{u_n(x,T)}{1+u_n(x,T)}\right)^2\right), \label{eq: Impulsive: intro: model discrete dispersal rain water uptake w between year}
	\end{align}
\end{subequations}
after dropping the tildes for brevity, where $u_n=u_n(x,t)$, $w_n=w_n(x,t)$, $x \in \R$, $0<t<T$ and $n \in \N$. While the dimensionless parameters $A$, $B$ and $T$ are combinations of several of the original parameters, they can be interpreted as the total amount of rainfall per year, rate of plant loss and time between separate rain and dispersal events, respectively. The water redistribution parameters $\nu$ and $d$ describe the ratio of the water flow coefficients (advection and diffusion, respectively) to the plant dispersal kernel width $1/a$. Their estimates are also included in Table \ref{tab: Implsive: Intro: Parameters}. 

In this form, $T=4$ corresponds to rain/dispersal events occurring once per year. Even though we present results for $0<T<4$ in this paper, it is important to emphasise that ecologically it is meaningless to consider the limit $T\rightarrow 0$. {\color{changes}In this limit, \eqref{eq: Impulsive: intro: model discrete dispersal rain water uptake} does not tend to an ecologically meaningful model with a continuous rainfall regime. Instead, it simply reduces to an integrodifference system given by \eqref{eq: Impulsive: intro: model discrete dispersal rain water uptake u between year} and \eqref{eq: Impulsive: intro: model discrete dispersal rain water uptake w between year}, in which no plant death and water evaporation occur.}

%%%%%%%%%%%%%%%%%%%%%%%%%%%%%%%%%%%%%%%%%%%%%%%%%%%%%%%%%%%%%%%%%%%%%%%%%%%%%%%%%%%%%%%%%%%%%%%%%%%%%%%%%%%%%%%%%%%%%%%%%%%
%%%%%%%%%%%%%%%%%%%%%%%%%%%%%%             IMPUSLIVE MODEL             %%%%%%%%%%%%%%%%%%%%%%%%%%%%%%%%%%%%%%%%%%%%%%%%%%%%%
%%%%%%%%%%%%%%%%%%%%%%%%%%%%%%%%%%%%%%%%%%%%%%%%%%%%%%%%%%%%%%%%%%%%%%%%%%%%%%%%%%%%%%%%%%%%%%%%%%%%%%%%%%%%%%%%%%%%%%%%%%% 

\section{Onset of patterns}\label{sec: Impulsive} 
A common method to study the onset of patterns is linear stability analysis. Spatial patterns occur if a steady state that is stable to spatially homogeneous perturbations becomes unstable if a spatially heterogeneous perturbation is introduced. In this section we apply such an approach to the impulsive model \eqref{eq: Impulsive: intro: model discrete dispersal rain water uptake} on flat ground for the Laplace kernel. Our analysis shows that while a smaller number of strong precipitation events inhibits their onset by decreasing the size of the parameter region supporting the onset of patterns, it also increases the requirements on the total amount of rainfall for plants to persist in a spatially uniform equilibrium. We further show that the introduction of temporal rainfall intermittency replaces water diffusion as the main cause of spatial patterns.

%%%%%%%%%%%%%%%%%%%%%%%%%%%%%%%%%%%%%%%%%%%%%%%%%%%%%%%%%%%%%%%%%%%%%%%%%%%%%%%%%%%%%%%%%%%%%%%%%%%%%%%%%%%%%%%%%%%%%%%%%%%
%%%%%%%%%%%%%%%%%%%%%%%%%%%%%%             LINEAR STABILITY             %%%%%%%%%%%%%%%%%%%%%%%%%%%%%%%%%%%%%%%%%%%%%%%%%%%
%%%%%%%%%%%%%%%%%%%%%%%%%%%%%%%%%%%%%%%%%%%%%%%%%%%%%%%%%%%%%%%%%%%%%%%%%%%%%%%%%%%%%%%%%%%%%%%%%%%%%%%%%%%%%%%%%%%%%%%%%%% 

\subsection{Linear Stability Analysis}\label{sec: Impuslsive: LinStab}

{\color{changes}
The use of linear stability analysis to determine conditions for the onset of patterns in a system concentrates on the calculation of growth/decay rates of perturbations to a spatially uniform equilibrium. In PDE systems and integrodifference systems, spatially uniform steady states are constant in both space and time, and can be calculated by setting all derivatives to zero (PDE systems) or imposing $u_{n+1} = u_n$ (integrodifference systems). By contrast, spatially uniform equilibria of impulsive systems are not constant in time. Instead, they are periodic in time with period $T$, the time between the occurrences of pulse-type events, and undergo the same cycle during each interpulse period. Consequently, time derivatives in the interpulse PDEs cannot be neglected in the calculation of spatially uniform equilibria. For the given impulsive model
\begin{subequations}
	\label{eq: Impulsive: intro: model general}
	\begin{align}
	\frac{\partial u_n}{\partial t} &= -Bu_n, \label{eq: Impulsive: intro: model general u in year}\\
	\frac{\partial w_n}{\partial t} &= -w_n + \nu \frac{\partial w_n}{\partial x} + d \frac{\partial^2 w_n}{\partial x^2}, \label{eq: Impulsive: intro: model general w in year}\\
	u_{n+1}(x,0) &=\tilde{f}\left(u_n(x,T),w_n(x,T)\right), \label{eq: Impulsive: intro: model general u between year}\\
	w_{n+1}(x,0) &= \tilde{g}\left(u_n(x,T),w_n(x,T)\right), \label{eq: Impulsive: intro: model general w between year}
	\end{align}
\end{subequations}
where $u_n=u_n(x,t)$, $w_n=w_n(x,t)$, $x \in \R$, $0<t<T$ and $n \in \N$, the assumption of spatial uniformity reduces the impulsive system to the difference system
\begin{subequations}\label{eq: Impulsive: model: difference model to find equilibria}
\begin{align}
u_{n+1}(0) &=\tilde{f}\left(u_n(0)e^{-BT},w_n(0)e^{-T}\right), \\
w_{n+1}(0) &= \tilde{g}\left(u_n(0)e^{-BT},w_n(0)e^{-T}\right),
\end{align}
\end{subequations}
after solving \eqref{eq: Impulsive: intro: model general u in year} and \eqref{eq: Impulsive: intro: model general w in year}, where the densities during any interpulse period are given by $u_n(t) = u_n(0)e^{-Bt}$ and $w_n(t) = w_n(0)e^{-t}$ for $0\le t\le T$. Even though a non-trivial equilibrium $(\overline{u}(t), \overline{w}(t))$ of \eqref{eq: Impulsive: intro: model general} is a periodic function of time, we introduce the notation $\overline{u}^0$ and $\overline{w}^0$ to denote the equilibrium densities at the start of the interpulse period, i.e. $\overline{u}^0:=\overline{u}(0)$ and $\overline{w}^0:=\overline{w}(0)$. This yields that the general, time-dependent equilibrium densities can be written as $\overline{u}(t)=\overline{u}^0e^{-Bt}$ and $\overline{w}=\overline{w}^0e^{-t}$ for $0\le t\le T$. For brevity, we use the notation $(\overline{u}^0,\overline{w}^0)$ to refer to the equilibrium in the analysis that follows.  Thus, from the reduced difference model \eqref{eq: Impulsive: model: difference model to find equilibria} it follows that the equilibria of \eqref{eq: Impulsive: intro: model general} can be found by solving
\begin{align*}
\overline{u}^0 &=\tilde{f}\left(\overline{u}^0e^{-BT},\overline{w}^0e^{-T}\right), \\
\overline{w}^0&= \tilde{g}\left(\overline{u}^0e^{-BT},\overline{w}^0e^{-T}\right).
\end{align*}
Application of this procedure to \eqref{eq: Impulsive: intro: model discrete dispersal rain water uptake} gives the spatially uniform equilibria of the impulsive system as
}
\begin{align*}
(\overline{u}^0_d,\overline{w}^0_d) &= \left(0, \frac{ATe^T}{e^T-1}\right), 
(\overline{u}^0_\pm,\overline{w}^0_\pm) = \left({\frac { \left( AT-2{{e}^{-T}}+2 \right) {{e}^{-BT}
		}\pm \sqrt {\eta}+2{{e}^{-T}}-2}{2{{e}^{-BT}} \left( 1 - {{e}^{-BT}} \right) }}, \right.\\  &\left. \frac { 2\left(  \left( AT-2{{ e}^{-T}}+1 \right) {{e}^{-BT}}\pm\sqrt {\eta}+2{{ e}^{-T}}-1 \right) \left(1- {{e}^{-BT}} \right) }{ \left(  \left( AT-2{{e}^{-T}}+2 \right) {{e}^{-BT}}\pm\sqrt{\eta}+2{{e}^{-T}}-2 \right) {{e}^{-BT}}} \right),
\end{align*}
where
\begin{multline*}
\eta = \left( 4 \left( {{e}^{-T}} \right) ^{2}+ \left( -4AT-4\right) {{e}^{-T}}+{A}^{2}{T}^{2}+4AT \right)  \left( {{e}^{-BT}} \right) ^{2} \\ +4 \left( {{ e}^{-T}}-1 \right)  \left( AT-2{{e}^{-T}} \right) {{e}^{-BT}}  +4 \left( {{e}^{-T}}\right) ^{2}-4{{e}^{-T}}.
\end{multline*}
The steady states $(\overline{u}^0_\pm,\overline{w}^0_\pm)$ only exist provided that 
\begin{align}\label{eq: Impulsive: LinStab: lower bound on A for ss to exist}
A> A_{\min}:= \frac{2\left(1-e^{-T}+\sqrt{1-e^{-T}}\right)\left(1-e^{-BT}\right)}{Te^{-BT}},
\end{align}
to ensure positivity of $\eta$. In principle, this structure is very similar to that of the Klausmeier models \eqref{eq: Intro Klausmeier local} and \eqref{eq: Intro Klausmeier nonlocal}. For rainfall levels below $A_{\min}$ only the desert steady state $(\overline{u}^0_d,\overline{w}^0_d)$ exists and plants die out, while for sufficiently large amounts of precipitation two further spatially uniform equilibria with non-zero vegetation density exist. An initial conclusion therefore is the existence of an inhibitory effect of long drought periods. The existence threshold $A_{\min}$ of $(\overline{u}^0_+,\overline{w}^0_+)$ increases with the interpulse time $T$ and thus enlarges the parameter region in which the desert equilibrium $(\overline{u}^0_d,\overline{w}^0_d)$ is the only spatially uniform steady state. {\color{changes}Even though $A_{\min}$ does not yield any information on the existence of spatially non-uniform solutions for low precipitation levels, we use this threshold as a proxy for the minimum water requirements of the ecosystem. This crucial property is revisited in our discussion on model extensions in Sec. \ref{sec: Impulsive: Simulations}.}

Similar to the Klausmeier models, spatial patterns arise from the steady state $(\overline{u}^0_+,\overline{w}^0_+)$ which is stable to spatially homogeneous perturbations (Proposition \ref{prop: Impulsive: LinStab: stability spatially homog pert}). 
The stability structure of the steady states of the Klausmeier models is preserved in the impulsive model, i.e. the desert steady state $(\overline{u}^0_d,\overline{w}^0_d)$ and the vegetation steady state $(\overline{u}^0_+,\overline{w}^0_+)$ are stable to spatially homogeneous perturbations, while the other vegetation steady state $(\overline{u}^0_-,\overline{w}^0_-)$ is unstable for all biologically realistic parameter choices.

\begin{prop}\label{prop: Impulsive: LinStab: stability spatially homog pert}
	Let $A>A_{\min}$,
	\begin{align*}
	\overline{B_2} = \frac{1}{T} \ln\left({\color{changes}1+\frac{ATe^T\sqrt{e^T-1}\left(\sqrt{e^T} - \sqrt{e^T-1}\right)}{2\left(e^T-1\right)}}\right),
	\end{align*}
	and $J_1(B):=e^{-T(B+1)} (\overline{\alpha} \overline{\delta} - \overline{\gamma} \overline{\beta}){\color{changes}-1}$, where
	\begin{align}\begin{split}\label{eq: Impulsive: LinStab: Jacobian coefficients spatially homogeneous}
	\overline{\alpha} &=\tilde{f}_u\left(\overline{u}^0e^{-BT},\overline{w}^0e^{-T}\right), \quad  \overline{\beta} = \tilde{f}_w\left(\overline{u}^0e^{-BT},\overline{w}^0e^{-T}\right), \\ \overline{\gamma} &= \tilde{g}_u\left(\overline{u}^0e^{-BT},\overline{w}^0e^{-T}\right), \quad \overline{\delta} = \tilde{g}_w\left(\overline{u}^0e^{-BT},\overline{w}^0e^{-T}\right).
	\end{split}\end{align}
	If $J_1(B)$ admits a positive real root $\overline{B_1}$, the steady state $(\overline{u}^0_+, \overline{w}^0_+)$ is stable to spatially homogeneous perturbations if $B<\min\{\overline{B_1},\overline{B_2}\}$ provided that $\overline{B_2} \in \R$ or $B<\overline{B_1}$ provided that $\overline{B_2} \notin \R$. If no positive real solution of $J_1(B)=0$ exists, then $(\overline{u}^0_+, \overline{w}^0_+)$ is stable if $B<\overline{B_2}$ provided that $\overline{B_2} \in \R$.
\end{prop}

{\color{changes}The proof of Proposition \ref{prop: Impulsive: LinStab: stability spatially homog pert}, as well as all those of all other propositions, is deferred until the end of the section.}

From Proposition \ref{prop: Impulsive: LinStab: stability spatially homog pert} it follows that the steady state $(\overline{u}^0_+,\overline{w}^0_+)$ is stable to spatially homogeneous perturbations close to $B=0$ {\color{changes} for biologically relevant parameters (i.e. $A,T>0$).}
Similar calculations yield that $(\overline{u}^0_-, \overline{w}^0_-)$ is unstable close to $B=0$. In particular, for $B=0.45$, the highest estimate of the plant mortality parameter (see Table \ref{tab: Implsive: Intro: Parameters}), the steady state $(\overline{u}^0_+,\overline{w}^0_+)$ is stable for all $(A,T)$ pairs with $A>A_{\min}$, while similarly  $(\overline{u}^0_-,\overline{w}^0_-)$ is unstable.

We investigate the existence of spatial patterns by introducing spatially heterogeneous perturbations to the steady state $(\overline{u}^0,\overline{w}^0):=(\overline{u}^0_+,\overline{w}^0_+)$. The following propositions provide conditions for a steady state to be stable to such spatially heterogeneous perturbations and yield results on the effects of rainfall intermittency on the onset of spatial patterns.

\begin{prop}\label{prop: Impulsive: LinStab: stability to spatially heterogeneous perturbations general}
	Let $\tilde{f}$ be of the form $\tilde{f}(u,w) = u+\phi *\tilde{f}_1(u,w)$. A steady state  $(\overline{u}^0,\overline{w}^0)$ of the impulsive model \eqref{eq: Impulsive: intro: model general} is stable to spatially heterogeneous perturbations if $|\lambda(k)|<1$ for both eigenvalues $\lambda \in \C$ of
	\begin{align}
	\label{eq: Impulsive: LinStab: Jacobian spatial heterogeneous}
	J = \mattwo{\left(1+\widehat{\phi}(k)\tilde{\alpha} \right)e^{-BT}}{\widehat{\phi}(k)\tilde{\beta} e^{-\left(1-i\nu k+dk^2\right)T}}{ \tilde{\gamma} e^{-BT}}{\tilde{\delta} e^{-\left(1-i\nu k+dk^2\right)T}},
	\end{align}
	for $k>0$, where 
	\begin{align}\label{eq: Impulsive: LinStab: linearisation coefficients}
	\begin{split}
	\tilde{\alpha} &= \frac{\partial \tilde{f}_1}{\partial u}\left(\overline{u}^0e^{-BT},\overline{w}^0e^{-T}\right), \quad \tilde{\beta} = \frac{\partial \tilde{f}_1}{\partial w}\left(\overline{u}^0e^{-BT},\overline{w}^0e^{-T}\right), \\
	\tilde{\gamma} &= \frac{\partial \tilde{g}}{\partial u}\left(\overline{u}^0e^{-BT},\overline{w}^0e^{-T}\right), \quad \tilde{\delta} = \frac{\partial \tilde{g}}{\partial w}\left(\overline{u}^0e^{-BT},\overline{w}^0e^{-T}\right).
	\end{split}
	\end{align}
\end{prop}
The entries of the Jacobian \eqref{eq: Impulsive: LinStab: Jacobian spatial heterogeneous} are complex-valued. However, a significant simplification is achieved by considering the model on flat ground, i.e. the case of $\nu =0$, thus allowing an application of the Jury criterion {\color{changes}(see e.g. \cite{Murray1989})} to determine conditions such that $|\lambda(k)|<1$ for both eigenvalues of the Jacobian.

\begin{prop}\label{prop: Impuslive: LinStab: stability to spatially heterogeneous perturbations general}
	The steady state $(\overline{u}^0,\overline{w}^0)$ of the impulsive model \eqref{eq: Impulsive: intro: model general} on flat ground is stable to spatially heterogeneous perturbations if 
	\begin{align}\label{eq: Impulsive: LinStab: jury conditions spatial 2}
	1+\det(J(k))-\tr(J(k))>0,
	\end{align}
	for all $k>0$, where $J$ is the Jacobian defined in \eqref{eq: Impulsive: LinStab: Jacobian spatial heterogeneous} with $\nu=0$.
\end{prop}

This provides a sufficient condition for the occurrence of spatial patterns. Both the local and the nonlocal Klausmeier models undergo a diffusion-driven instability on flat ground for any level of rainfall, meaning that a sufficiently large ratio of water diffusion rate to plant diffusion rate yields a pattern-inducing instability (see Figure \ref{fig: difference: LinStab: nonlocal Klausmeier in Ad plane}). This is not the case for the impulsive model. For sufficiently high levels of rainfall, patterns cannot occur for any level of the diffusion coefficient $d$, the ratio of the water diffusion rate to the plant dispersal kernel width. It is indeed the time $T$ between rainfall pulses that determines for which levels of precipitation patterns can form. Only for smaller values of $A$ an increase of diffusion through the critical value $d_c(A)$ causes an instability and thus the onset of patterns. Reverting back to parameters in dimensional form, this also shows that for sufficiently low precipitation levels, wider plant dispersal kernels inhibit the onset of patterns, which is in agreement with results from the nonlocal Klausmeier model \eqref{eq: Intro Klausmeier nonlocal} \cite{Eigentler2018nonlocalKlausmeier}. Similar to the Klausmeier models, diffusion levels close to $d=0$ do not yield an instability for any rainfall parameters and there is a direct transition from the stable plant steady state to the desert steady {\color{changes}state} as $A$ is decreased through the lower bound $A_{\min}$. This is a conclusion of a numerical investigation (Figure \ref{fig: Impulsive: LinStab: jury2 in Ad parameter plane flat ground}) of the stability condition \eqref{eq: Impulsive: LinStab: jury conditions spatial 2} using the Laplace kernel \eqref{eq:difference equations: simulations: Laplacian} in the $A$-$d$ parameter plane. This analysis was performed for various different choices of the parameters $B$ and $T$ without showing any qualitative differences.

\begin{figure}
	\centering
	\subfloat[Impulsive model\label{fig: Impulsive: LinStab: jury2 in Ad parameter plane flat ground B045 T1}]{\includegraphics[width=0.48\textwidth]{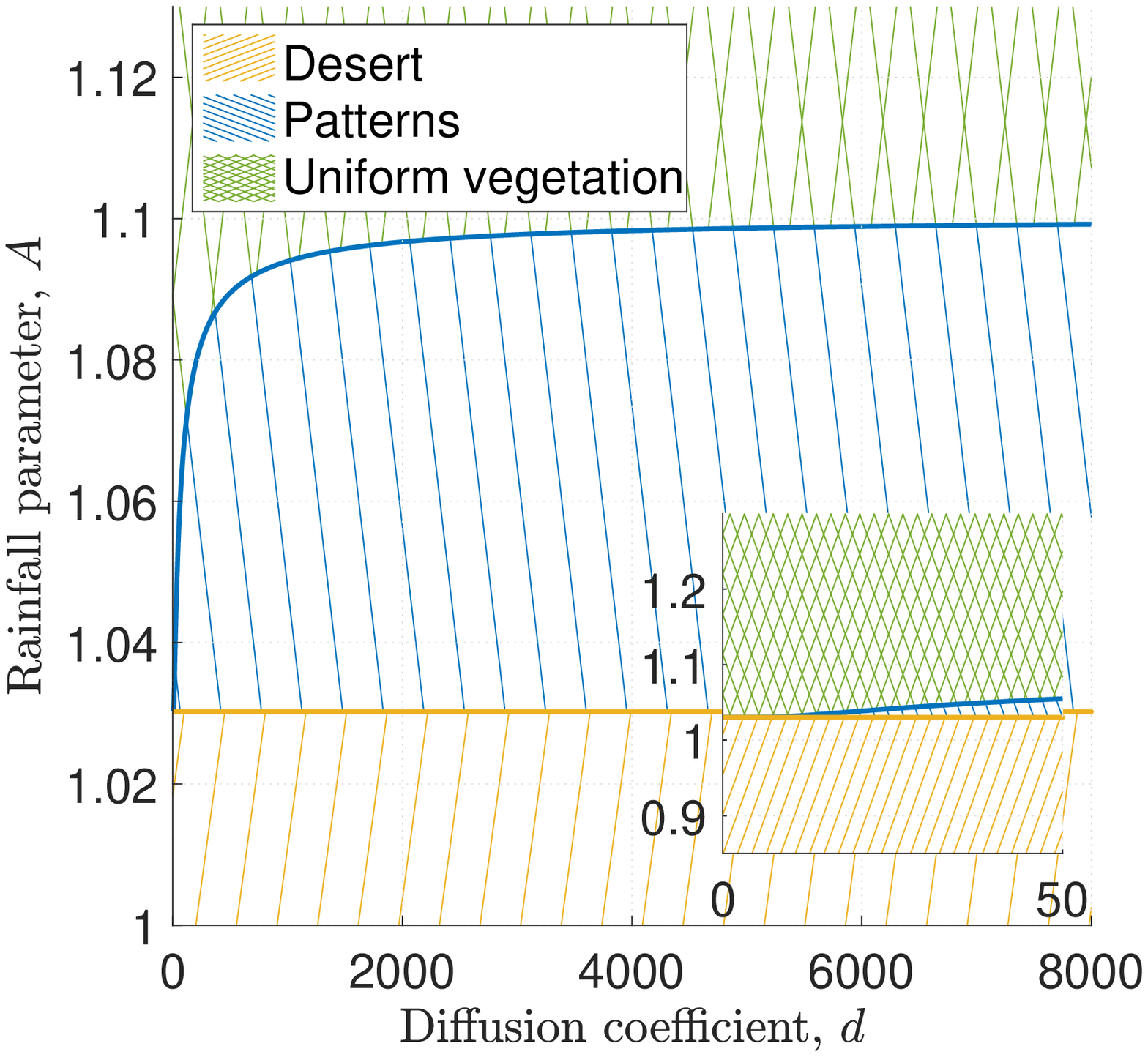}}
	\subfloat[Nonlocal Klausmeier model \label{fig: difference: LinStab: nonlocal Klausmeier in Ad plane}]{\includegraphics[width=0.48\textwidth]{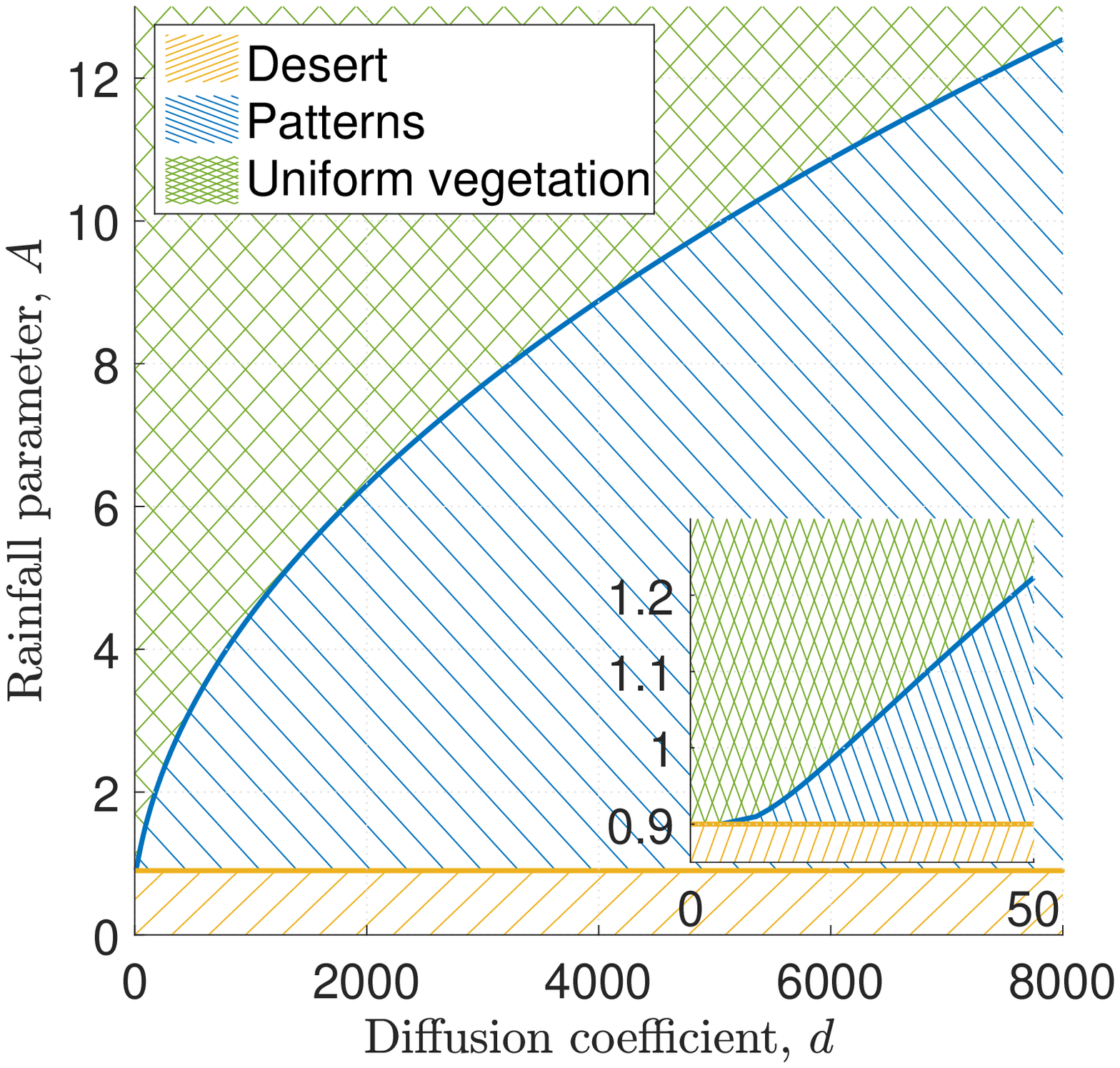}}
	\caption{The stability criterion \eqref{eq: Impulsive: LinStab: jury conditions spatial 2} in the $A$-$d$ parameter plane. Part (a) visualises where the second Jury condition \eqref{eq: Impulsive: LinStab: jury conditions spatial 2} changes sign and thus yields an instability, which causes the onset of spatial patterns (blue line). Given some value of $d$, the value of $A$ at which a transition between positivity and negativity of the condition occurs, is determined up to an interval of length $10^{-10}$. The level of $d$ is increased in variable increments. For $d$ close to 0, the increment is chosen to be $\Delta d = 0.1$, which then increases up to $\Delta d = 100$ as the value of $d$ increases. For any given $(A,d)$ pair, \eqref{eq: Impulsive: LinStab: jury conditions spatial 2} is evaluated for $k>0$ at increments of $\Delta k = 0.01$ until a value of $k_c$ is found for which the Jury condition is either negative or positive and increasing. In the former case, the $(A,d)$ pair supports the onset of spatial patterns, in the latter the interval $[k_c-\Delta k,k_c]$ is investigated further with smaller increments in $k$. If still no $k$ is found for which the Jury condition is negative it is assumed that \eqref{eq: Impulsive: LinStab: jury conditions spatial 2}  is not satisfied. The other parameter values used in this analysis were $T=0.5$ and $B=0.45$, with $\phi$ being the Laplacian kernel \eqref{eq:difference equations: simulations: Laplacian}. {\color{changes}This yields $A_{\min} = 1.03$, and $A_{\max} = 1.099$.} A comparison to the nonlocal Klausmeier model, which undergoes a diffusion-driven instability, is shown in (b). {\color{changes}Here $A_{\min} = 2B = 0.9$.} The insets show the behaviour close to $d=0$.}\label{fig: Impulsive: LinStab: jury2 in Ad parameter plane flat ground}
\end{figure}

The evaluation of \eqref{eq: Impulsive: LinStab: jury conditions spatial 2} in the $A$-$d$ parameter plane suggests a closer investigation of the stability condition \eqref{eq: Impulsive: LinStab: jury conditions spatial 2} for $d\rightarrow \infty$ (Proposition \ref{prop: Impulsive: LinStab: Amax in d infinity limit}) and $A=A_{\min}$ (Proposition \ref{prop: Impulsive: LinStab: dc in A=Amin case flat ground}). The former provides information on the level of rainfall $A_{\max}$ above which no instability can occur, while the latter yields the locus of $d_{A_{\min}}$, the minimum value of diffusion required for an instability to occur.

\begin{prop}\label{prop: Impulsive: LinStab: Amax in d infinity limit}
	If $d\rightarrow \infty$ in the impulsive model \eqref{eq: Impulsive: intro: model general} on flat ground with the Laplace kernel \eqref{eq:difference equations: simulations: Laplacian}, then $(\overline{u}^0,\overline{w}^0)$ is unstable to spatially heterogeneous perturbations if $A<A_{\max}$, where $A_{\max}$ satisfies
	\begin{align}\label{eq: Impulsive: LinStab: Jury2 condition flat ground d limit for all k}
	\left(1+\tilde{\alpha}\left(A_{\max}\right)\right)e^{-BT}-1 {\color{changes}= } 0.
	\end{align}
\end{prop}

\begin{cor}\label{cor: Impulsive: LinStab: decrease of pattern supporting interval d infinity limit}
	The relative size $(A_{\max}-A_{\min})/A_{\min}$ of the interval $[A_{\min},A_{\max}]$ is proportional to $e^{-2T}$ as $T\rightarrow \infty$.	
\end{cor}

Given a set of parameters $(B,T)$, \eqref{eq: Impulsive: LinStab: Jury2 condition flat ground d limit for all k} can be solved numerically to provide the level of rainfall $A$ at which a transition between uniform and patterned vegetation occurs in the limit $d\rightarrow \infty$. In combination with the preceding results, this is the threshold $A_{\max}$ beyond which no pattern onset can occur. This is in stark contrast {\color{changes}to} the classical case of a diffusion-driven instability which occurs in the Klausmeier models \eqref{eq: Intro Klausmeier local} and \eqref{eq: Intro Klausmeier nonlocal} for which $A_{\max}\rightarrow \infty$ as $d\rightarrow\infty$. Together with the lower bound on the rainfall parameter $A_{\min}$ this allows a classification of the $T$-$A$ parameter plane into three regions (Figure \ref{fig: Impuslive: LinStab: dlimit AT plane full}); one in which the desert steady state is the only spatially uniform equilibrium to exist, one in which instability of the uniform plant steady state to spatially heterogeneous perturbations causes the onset of spatial patterns, and one in which any perturbations of the equilibrium $(\overline{u}^0,\overline{w}^0)$ decay and no pattern onset occurs. This classification of the $T$-$A$ parameter plane is based on the preceding linear stability analysis and the perturbation of the spatially uniform equilibrium $(\overline{u}^0,\overline{w}^0)$. This results in a classification that provides information regarding the onset of patterns but does not yield any knowledge of the existence of patterns away from their onset. While no systematic study of the whole parameter space using numerical continuation was performed, patterns for parameters outside the interval given by the linear stability analysis can be observed by slowly increasing/decreasing the rainfall parameter $A$ beyond/below the pattern onset-supporting interval when the system is already in a patterned state.

Proposition \ref{prop: Impulsive: LinStab: Amax in d infinity limit} indicates that a decrease in the frequency of precipitation events requires a higher amount of rainfall to avoid an instability. This does not mean that pattern onset occurs for a larger parameter range as periods of droughts become longer, as an increase in $T$ also increases the lower bound on the rainfall for the vegetation steady state to exist. Indeed, Corollary \ref{cor: Impulsive: LinStab: decrease of pattern supporting interval d infinity limit} provides information on the size of the interval for the rainfall parameter $A$ that supports the onset of patterns relative to the lower bound \eqref{eq: Impulsive: LinStab: lower bound on A for ss to exist} on the rainfall (Figure \ref{fig: Impuslive: LinStab: dlimit A interval size}). For small values of $T$ the size of this interval is larger than the lower bound on the rainfall, for larger $T$ the size of the pattern onset-supporting interval of rainfall levels decreases at a rate proportional to $e^{-2T}$.

\begin{figure}
	\centering
	\subfloat[\label{fig: Impuslive: LinStab: dlimit AT plane full}]{\includegraphics[width=0.48\textwidth]{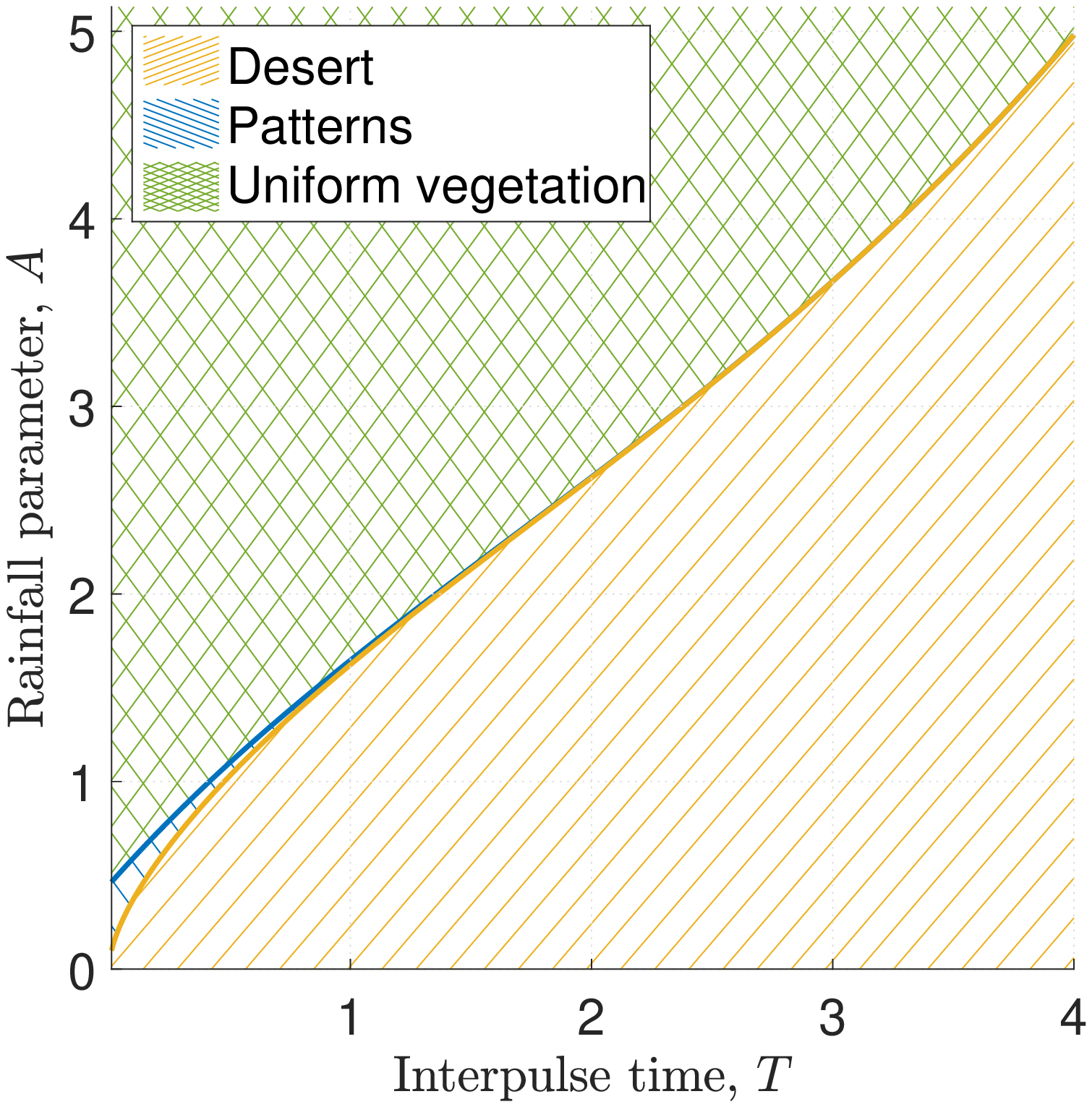}}
	\subfloat[\label{fig: Impuslive: LinStab: dlimit A interval size}]{\includegraphics[width=0.48\textwidth]{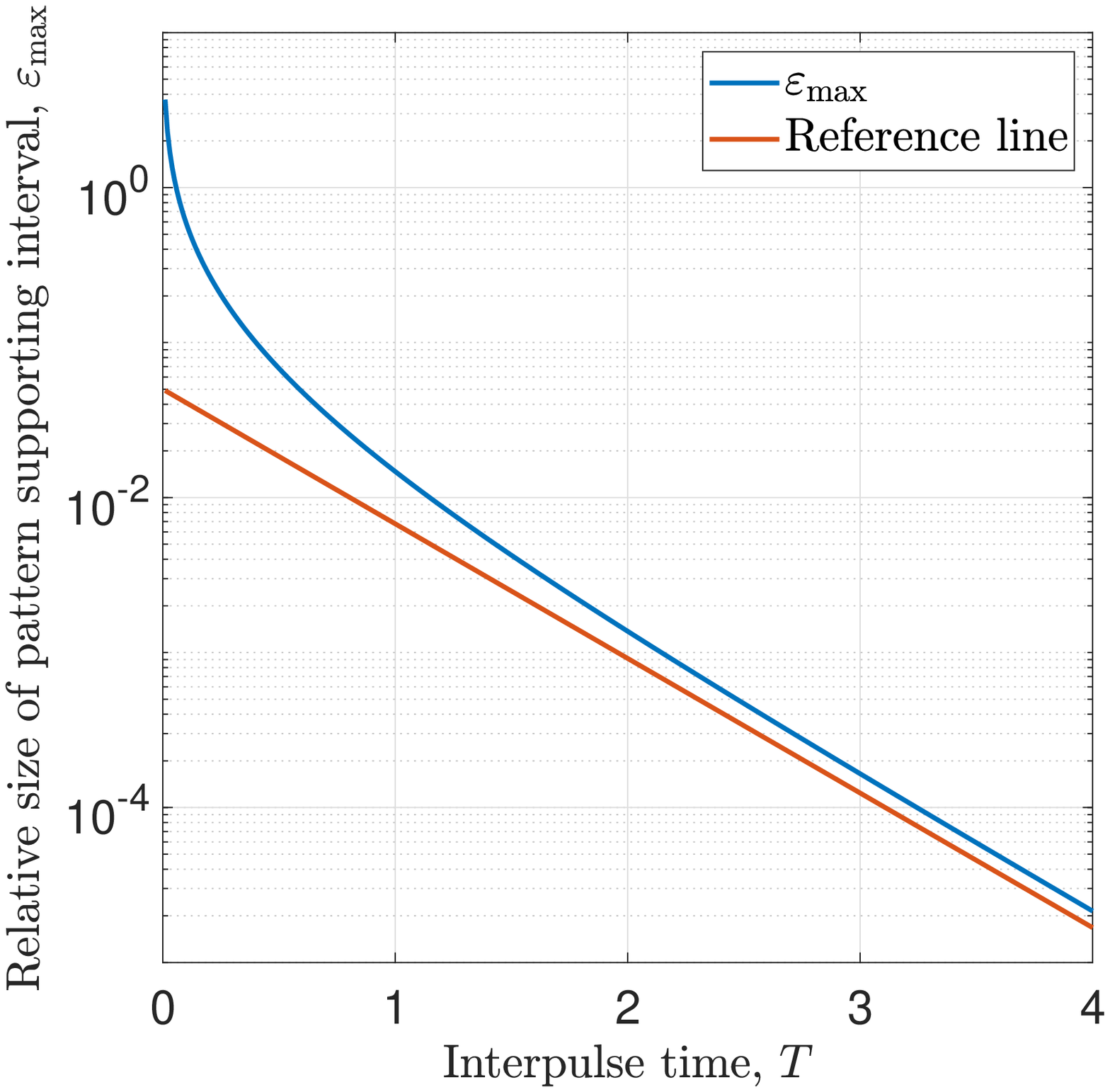}}
	\caption{Classification of the $A$-$T$ parameter plane and the relative size of the rainfall interval supporting pattern onset as $d\rightarrow \infty$. Part (a) shows a classification of the $A$-$T$ parameter plane in the limit $d \rightarrow \infty$ into regions in which uniform vegetation is stable, in which pattern onset occurs, and in which the desert state is the only spatially uniform equilibrium of the system. The transition $A=A_{\max}$ between uniform and patterned vegetation (blue line) is obtained by numerically solving \eqref{eq: Impulsive: LinStab: Jury2 condition flat ground d limit for all k}, while the lower bound $A=A_{\min}$ on the parameter region supporting pattern onset (yellow line) is obtained from the analytic condition \eqref{eq: Impulsive: LinStab: lower bound on A for ss to exist}. {\color{changes}The relative size $ \varepsilon_{\max}:= A_{\max}-A_{\min}$ of the parameter region supporting pattern onset is visualised in (b)} and is compared to a reference line of slope $\exp(-2T)$. The plant loss parameter is $B=0.45$. Note the logarithmic scale in (b).}\label{fig: Impuslive: LinStab: dlimit AT plane}
\end{figure}

For $A=A_{\min}$, the previous analysis (Figure \ref{fig: Impulsive: LinStab: jury2 in Ad parameter plane flat ground B045 T1}) suggests the existence of a threshold $d_{A_{\min}}$ on the diffusion coefficient $d$ below which no instability occurs. Similar to the Klausmeier models \eqref{eq: Intro Klausmeier local} and \eqref{eq: Intro Klausmeier nonlocal} this corresponds to a direct transition between the spatially uniform vegetation state and the desert state as the rainfall parameter $A$ decreases through $A_{\min}$.

\begin{prop}\label{prop: Impulsive: LinStab: dc in A=Amin case flat ground}
	If $A=A_{\min}$, there exists a threshold $ d_{A_{\min}}>0$ such that $(\overline{u}^0,\overline{w}^0)$ is unstable to spatially heterogeneous perturbations for $d>d_{A_{\min}}${\color{changes}.}
	%	, where $d_{A_{\min}}$ satisfies
	%\begin{align}\label{eq: Impulsive: LinStab: jury2 condition for A=Amin}
	% d_{A_{\min}} = \max \left\{d\ge0:\zeta\left(d, T, k\right) >0, \quad \text{for all} \quad k>0\right\},
	%\end{align}
	%i.e. is independent of the plant mortality $B$.
\end{prop}

The threshold given by Proposition \ref{prop: Impulsive: LinStab: dc in A=Amin case flat ground} is independent of the plant loss parameter $B$. Plant mortality does, however, affect $A_{\min}$, the level of rainfall required for the plant steady state to exist. The simplification provided by setting $A=A_{\min}$ is not sufficient to determine the threshold $d_{A_{\min}}$ on the diffusion coefficient explicitly, but similar to the analysis in the $d\rightarrow \infty$ case, it can be determined numerically for a given set of parameters. The results of this show that an increase in the time interpulse $T$ causes an increase in the threshold $d_{A_{\min}}$ on the diffusion coefficient, i.e. a higher ratio of water diffusion to plant dispersal kernel width is required to cause an instability leading to the onset of patterns.

%%%%%%%%%%%%%%%%%%%%%%%%%%%%%%%%%%%%%%%%%%%%%%%%%%%%%%%%%%%%%%%%%%%%%%%%%%%%%%%%%%%%%%%%%%%%%%%%%%%%%%%%%%%%%%%%%%%%%%%%%%%
%%%%%%%%%%%%%%%%%%%%%%%%%%%%%%    PROOFS             %%%%%%%%%%%%%%%%%%%%%%%%%%%%%%%%%%%%%%%%%%%%%%%%%%%%%%%%%%%%%%%%%%%%%%
%%%%%%%%%%%%%%%%%%%%%%%%%%%%%%%%%%%%%%%%%%%%%%%%%%%%%%%%%%%%%%%%%%%%%%%%%%%%%%%%%%%%%%%%%%%%%%%%%%%%%%%%%%%%%%%%%%%%%%%%%%% 

\begin{proof}[Proof of Proposition \ref{prop: Impulsive: LinStab: stability spatially homog pert}]
	Linear stability analysis of a steady state $(\overline{u}^0,\overline{w}^0)$ of the impulsive model \eqref{eq: Impulsive: intro: model general} in a spatially uniform setting is equivalent to linear stability analysis of the difference system {\color{changes} \eqref{eq: Impulsive: model: difference model to find equilibria} with $\tilde{f}(u,w) = u+(u/(1+u))^2(w+TA)$ and $\tilde{g}(u,w) = (w+TA)(1-(u/(1+u))^2)$. Linearisation about the steady state and introduction} of a perturbation proportional to $\lambda^n$ yields that the growth factor $\lambda \in \C$ is an eigenvalue of the Jacobian
	\begin{align*}
	J(\overline{u}^0,\overline{w}^0) = \mattwo{e^{-BT}  \overline{\alpha}}{e^{-T}\overline{\beta}}{e^{-BT} \overline{\gamma}}{e^{-T} \overline{\delta}}.
	\end{align*}
	The Jury conditions then yield stability of a steady state $(\overline{u}^0,\overline{w}^0)$ if 
	\begin{subequations}
		\label{eq: Impulsive: LinStab: Jury conditions}
		\begin{align}
		e^{-T(B+1)} \left(\overline{\alpha} \overline{\delta} - \overline{\gamma} \overline{\beta} \right) &<1, \label{eq: Impulsive: LinStab: Jury conditions 1}\\
		1+e^{-T(B+1)}  \left(\overline{\alpha} \overline{\delta} - \overline{\gamma} \overline{\beta} \right)  &> \left|e^{-BT} \overline{\alpha} +e^{-T} \overline{\delta} \right|, \label{eq: Impulsive: LinStab: Jury conditions 2}
		\end{align}
	\end{subequations}
	are both satisfied.

	The first Jury condition \eqref{eq: Impulsive: LinStab: Jury conditions 1} yields $J_1(B)<0$.	
	For $(\overline{u}^0,\overline{w}^0)=(\overline{u}^0_+, \overline{w}^0_+)$, $J_1(0)=-1$ and thus the condition is satisfied for $B<\overline{B_1}$, where $\overline{B_1}$ is the smallest real positive root of $J_1(B)$ provided it exists. 
	The second Jury condition \eqref{eq: Impulsive: LinStab: Jury conditions 2} is $J_2(B):= 1 + e^{-T(B+1)}(\overline{\alpha}\overline{\delta}-\overline{\beta}\overline{\gamma}) -e^{-BT}\overline{\alpha} - e^{-T}\overline{\delta} >0$. For $(\overline{u}^0,\overline{w}^0)=(\overline{u}^0_+, \overline{w}^0_+)$, $J_2(0)=0$ and $\dif J_2/\dif B(0) = T>0$ and thus the condition is satisfied for $B<\overline{B_2}$, which is the smallest real positive root of $J_2(B)$ provided it exists.
\end{proof}

{\color{changes}\begin{proof}[Proof of Proposition \ref{prop: Impulsive: LinStab: stability to spatially heterogeneous perturbations general}]
		Similar to the Proposition \ref{prop: Impulsive: LinStab: stability spatially homog pert}, this proof is based on a linear stability analysis. Unlike in the proof of Proposition \ref{prop: Impulsive: LinStab: stability spatially homog pert}, the system cannot be immediately reduced to a difference system. Additionally, the convolution in \eqref{eq: Impulsive: intro: model general u between year} adds a complication. However, both these issues can be addressed by performing the analysis in Fourier space.

	As is standard with linear stability analysis, we investigate the behaviour of perturbations $( \tilde{u}(x,t), \tilde{w}(x,t))$ to a spatially uniform equilibrium $(\overline{u}(t),\overline{w}(t))=(\overline{u}^0e^{-Bt}, \overline{w}^0e^{-t})$ by setting
	\begin{align}
	\label{eq: Impuslive: LinStab: perturbation from steady state}
	u_n(x,t) = \overline{u}^0e^{-Bt} + \tilde{u}(x,t) \quad \text{and} \quad w_n(x,t) = \overline{w}^0e^{-t} + \tilde{w}(x,t).
	\end{align}
	 Substitution into the update equations \eqref{eq: Impulsive: intro: model general u between year} and \eqref{eq: Impulsive: intro: model general w between year} and linearisation yields
	\begin{subequations}
		\label{eq: Impulsive: LinStab: tilde eq before fourier}
		\begin{align}
		\tilde{u}_{n+1}(x,0) &= \tilde{u}_n(x,T) + \phi * \left(\tilde{\alpha} \tilde{u}_n(\cdot,T) + \tilde{\beta} \tilde{w}_n(\cdot,T) \right), \\
		\tilde{w}_{n+1} (x,0) &= \tilde{\gamma}\tilde{u}_n(x,T) + \tilde{\delta} \tilde{w}_n(x,T),
		\end{align}
	\end{subequations}
	noting that $\tilde{g}(\overline{u}^0e^{-BT},\overline{w}^0e^{-T}) = \overline{w}^0$ and $\tilde{f}_1(\overline{u}^0e^{-BT},\overline{w}^0e^{-T}) = \overline{u}^0(1-e^{-BT})$ by the definition of the spatially uniform equilibria. The Fourier transform applied to \eqref{eq: Impulsive: LinStab: tilde eq before fourier} gives
	\begin{subequations}
		\label{eq: Impulsive: LinStab: fourier interm}
		\begin{align}
		\widehat{\tilde{u}_{n+1}}(k,0) &= \left(1+\widehat{\phi}(k)\tilde{\alpha} \right) \widehat{\tilde{u}_{n}}(k,T) + \widehat{\phi}(k)\tilde{\beta} \widehat{\tilde{w}_{n}}(k,T), \\
		\widehat{\tilde{w}_{n+1}}(k,0) &= \tilde{\gamma} \widehat{\tilde{u}_{n}}(k,T) + \tilde{\delta} \widehat{\tilde{w}_{n}}(k,T),
		\end{align}
	\end{subequations}
	making use of the convolution theorem. The functions $\widehat{\tilde{u}_{n}}$ and $\widehat{\tilde{w}_{n}}$ satisfy the interpulse PDEs \eqref{eq: Impulsive: intro: model general u in year} and \eqref{eq: Impulsive: intro: model general w in year}. Taking the Fourier transform of the interpulse PDEs \eqref{eq: Impulsive: intro: model general u in year} and \eqref{eq: Impulsive: intro: model general w in year} gives 
	\begin{align*}
	\frac{\partial \widehat{\tilde{u}_n} (k,t)}{\partial t} = -B \widehat{\tilde{u}_n} (k,t), \quad
	\frac{\partial  \widehat{\tilde{w}_n} (k,t)}{\partial t} = - (1-i\nu k+dk^2) \widehat{\tilde{w}_n} (k,t), 
	\end{align*}
	which can be solved to 
	\begin{align}\label{eq: Impulsive: LinStab: fourier sol of pdes}
	\widehat{\tilde{u}_n} (k,t) =  \widehat{\tilde{u}_n} (k,0) e^{-Bt}, \quad \widehat{\tilde{w}_n} (k,t) = \widehat{\tilde{w}_n} (k,0) e^{-(1-i\nu k+dk^2)t}.
	\end{align}
	
	Substitution into \eqref{eq: Impulsive: LinStab: fourier interm} yields
	\begin{align*}
	\widehat{\tilde{u}_{n+1}}(k,0) &= \left(1+\widehat{\phi}(k)\tilde{\alpha} \right)e^{-BT} \widehat{\tilde{u}_{n}}(k,0) + \widehat{\phi}(k)\tilde{\beta} e^{-\left(1-i\nu k+dk^2\right)T}\widehat{\tilde{w}_{n}}(k,0), \\
	\widehat{\tilde{w}_{n+1}}(k,0) &= \tilde{\gamma} e^{-BT} \widehat{\tilde{u}_{n}}(k,0) + \tilde{\delta} e^{-\left(1-i\nu k+dk^2\right)T} \widehat{\tilde{w}_{n}}(k,0).
	\end{align*}
	from \eqref{eq: Impulsive: LinStab: fourier interm}. This is a linear difference system to which standard tools of stability analysis can be applied. In other words, the assumption that the perturbations $\widehat{\tilde{u}_n}$ and $\widehat{\tilde{w}_n}$ are proportional to $\lambda^n$ yields that the growth factor $\lambda \in \C$ is an eigenvalue of the Jacobian $J$.
\end{proof}}

\begin{proof}[Proof of Proposition \ref{prop: Impuslive: LinStab: stability to spatially heterogeneous perturbations general}]
	To investigate a steady state's stability on flat ground, the Jury conditions can be used. An instability occurs, if at least one of 
	\begin{subequations}
		\label{eq: Impulsive: LinStab: jury conditions spatial}
		\begin{align}
		\det(J)-1&<0, \label{eq: Impulsive: LinStab: jury conditions spatial 1}\\
		1+\det(J)-|\tr(J)| &> 0, 
		\end{align}
	\end{subequations}
	is not satisfied for some $k>0$. The first Jury condition \eqref{eq: Impulsive: LinStab: jury conditions spatial 1} is automatically satisfied due to stability to spatially homogeneous perturbations, because
	\begin{multline*}
	\det(J) = e^{-T\left(B+1+dk^2\right)} \left(\left(1+\widehat{\phi}(k)\tilde{\alpha} \right) \tilde{\delta} - \widehat{\phi}(k)\tilde{\beta}\tilde{\gamma} \right) = e^{-T\left(B+1+dk^2\right)} \left(\tilde{\delta} +\widehat{\phi}(k)\left(\tilde{\alpha}\tilde{\delta}  - \tilde{\beta}\tilde{\gamma} \right)\right) \\
	< e^{-T\left(B+1\right)} \left(\tilde{\delta} +\tilde{\alpha}\tilde{\delta}  - \tilde{\beta}\tilde{\gamma} \right) =  e^{-T\left(B+1\right)} \left(\overline{\alpha} \overline{\delta}  - \overline{\beta}\overline{\gamma} \right) <1,
	\end{multline*}
	for all $k>0$, noting that $1+\tilde{\alpha}  = \overline{\alpha}$, $\tilde{\beta} = \overline{\beta}$, $\tilde{\gamma} = \overline{\gamma}$ and $\tilde{\delta} = \overline{\delta}$, where $\overline{\alpha}$, $\overline{\beta}$, $\overline{\gamma}$ and $\overline{\delta}$ are defined in \eqref{eq: Impulsive: LinStab: Jacobian coefficients spatially homogeneous}. The last inequality makes use of the steady state's stability to spatially homogeneous perturbations, which in particular guarantees that \eqref{eq: Impulsive: LinStab: Jury conditions 1} holds. Therefore, assuming a steady state's stability to spatially homogeneous perturbations, a sufficient condition for spatial patterns to occur is the existence of some wavenumber $k>0$ such that the second Jury condition \eqref{eq: Impulsive: LinStab: jury conditions spatial 2} does not hold. In the case of model \eqref{eq: Impulsive: intro: model discrete dispersal rain water uptake} this condition can be slightly simplified by noting that $\tilde{\alpha}>0$ and $\tilde{\delta}>0$ and therefore $\tr(J)>0$ for all $k>0$. The condition thus becomes $1+\det(J)-\tr(J)>0$. 
\end{proof}

\begin{proof}[Proof of Proposition \ref{prop: Impulsive: LinStab: Amax in d infinity limit}]
	If $d \rightarrow \infty$ and $\nu=0$, then the Jacobian \eqref{eq: Impulsive: LinStab: Jacobian spatial heterogeneous} becomes
	\begin{align*}
	%\label{eq: Impulsive: LinStab: Jacobian spatial heterogeneous d limit}
	J = \mattwo{\left(1+\widehat{\phi}(k)\tilde{\alpha} \right)e^{-BT}}{0}{ \tilde{\gamma} e^{-BT}}{0}.
	\end{align*}
	Its determinant is clearly zero and therefore the stability condition simplifies to $1- (1+\widehat{\phi}(k) \tilde{\alpha})e^{-BT} >0$, for all $k>0$. For the Laplacian kernel \eqref{eq:difference equations: simulations: Laplacian} this is a polynomial in $k^2$, which after rearranging becomes
	\begin{align}\label{eq: Impulsive: LinStab: Jury2 condition flat ground d limit Laplace}
	k^2> \frac{\left(\left(1+\tilde{\alpha}\right)e^{-BT}-1\right)}{1-e^{-BT}}.
	\end{align}
	Stability of the steady state requires \eqref{eq: Impulsive: LinStab: Jury2 condition flat ground d limit Laplace} to hold for all $k>0$. This is only possible if the right hand side of \eqref{eq: Impulsive: LinStab: Jury2 condition flat ground d limit Laplace} is negative. Thus, an instability causing the onset of spatial patterns occurs if 
	\begin{align}\label{eq: Impulsive: LinStab: Stab criterion}
	\left(1+\tilde{\alpha}\right)e^{-BT}-1 > 0.
	\end{align} {\color{changes}The coefficient $\tilde{\alpha}$ is decreasing in $A$ and thus there exists a threshold $A=A_{\max}$ such that an instability occurs for all $A<A_{\max}$.}
\end{proof}

\begin{proof}[Proof of Corollary \ref{cor: Impulsive: LinStab: decrease of pattern supporting interval d infinity limit}]
	Substitution of $A=A_{\min}(1+\varepsilon)$ into \eqref{eq: Impulsive: LinStab: Stab criterion} gives
	\begin{align*}
	\varepsilon < \varepsilon_{\max}:=\frac{1}{8\left(e^{\frac{3T}{2}}\sqrt{e^T-1}+e^{2T} - e^{\frac{T}{2}}\sqrt{e^T-1}-e^T \right)},
	\end{align*}
	after linearisation in $\varepsilon$. The right hand side $\varepsilon_{\max}$ denotes the relative size of the rainfall interval supporting pattern onset. Its logarithm decreases at rate
	\begin{align*}
	\left(\ln\left(\varepsilon_{\max} \right)\right)' = -\frac{4e^{\frac{3T}{2}}\sqrt{e^T-1}+4e^{2T} - 2e^{\frac{T}{2}}\sqrt{e^T-1}-5e^T+1}{2\sqrt{e^T-1}\left(e^T\sqrt{e^T-1}+e^{\frac{3T}{2}} - \sqrt{e^T-1} -e^{\frac{T}{2}}\right)} \rightarrow -2 \quad \text{as} \quad T\rightarrow \infty.
	\end{align*}
	This shows the exponential decay of the relative interval size $\varepsilon_{\max}$.
\end{proof}

\begin{proof}[Proof of Proposition \ref{prop: Impulsive: LinStab: dc in A=Amin case flat ground}]
	Setting $A=A_{\min}$ provides a significant simplification as the equilibrium becomes
	\begin{align*}
	\left( \overline{u}^0,\overline{w}^0\right) = \left(\sqrt{1-e^{-T}}e^{BT}, \frac{\left(e^{BT}-1\right)\left(1+2\sqrt{1-e^{-T}}\right)}{\sqrt{1-e^{-T}}} \right).
	\end{align*}
	Thus the coefficients $\tilde{\alpha}$, $\tilde{\beta}$, $\tilde{\gamma}$ and $\tilde{\delta}$ given by \eqref{eq: Impulsive: LinStab: linearisation coefficients} become
	\begin{align*}
	\tilde{\alpha}_{A_{\min}} &=\frac{2\left(e^{BT}-1\right)\left(2-e^{-T}+2\sqrt{1-e^{-T}}\right)}{\left(1+\sqrt{1-e^{-T}}\right)^3}, \quad
	&\tilde{\beta}_{A_{\min}} = \frac{1-e^{-T}}{\left(1+\sqrt{1-e^{-T}}\right)^2}, \\
	\tilde{\gamma}_{A{\min}} &= -\tilde{\alpha}_{A_{\min}}, \quad
	&\tilde{\delta}_{A_{\min}} = \frac{2\sqrt{1-e^{-T}}+1}{\left(1+\sqrt{1-e^{-T}}\right)^2},
	\end{align*}
	respectively. The Jacobian \eqref{eq: Impulsive: LinStab: Jacobian spatial heterogeneous} then is 
	\begin{align*}
	J_{A_{\min}} =  \mattwo{\left(1+\widehat{\phi}(k)\tilde{\alpha}_{A_{\min}} \right)e^{-BT}}{\widehat{\phi}(k)\tilde{\beta}_{A_{\min}} e^{-\left(1+dk^2\right)T}}{ \tilde{\gamma}_{A_{\min}} e^{-BT}}{\tilde{\delta}_{A_{\min}} e^{-\left(1+dk^2\right)T}},
	\end{align*}
	and hence the steady state $(\overline{u}^0,\overline{w}^0)$ is stable to spatially heterogeneous perturbations if
	\begin{align}\label{eq: Impulsive: LinStab: Amin case stability condition}
	1+\det\left( J_{A_{\min}}\right) - \tr\left(J_{A_{\min}}\right) = \zeta \left(1-e^{-BT}\right) >0 \Longleftrightarrow \zeta >0, \quad \text{for all} \quad k>0
	\end{align}
	where 
	\begin{multline*}
	\zeta  = \frac{1}{{ \left( {{ e}^{T/2}}+\sqrt {{{ e}^{T}}-	1} \right) ^{3}}} \left( \left(  \left( -2\tilde{\beta}_{A_{{\min}}}-2\tilde{\delta}_{A_{{\min}}} \right) \widehat{\phi}(k)+3\tilde{\delta}_{A_{{\min}}} \right) {{ e}^{-Td{k}	^{2}-T/2}} \right.  \\ \left. + \left(  \left( 4\tilde{\beta}_{A_{{\min}}}+4\tilde{\delta}_{A_{{\min}}} \right) \widehat{\phi}(k)  - 4\tilde{\delta}_{A_{{\min}}} \right) {{ e}^{-Td{k}^{2}+T/2}}+\sqrt {{{ e}^{T}}-1}{{ e}^{- \left( d{k}^{2}+1	\right) T}}\tilde{\delta}_{A_{{\min}}}  \right. \\ \left. +4\sqrt {{{ e}^{T}}-1} \left(\left( \tilde{\beta}_{A_{{\min}}}+\tilde{\delta}_{A_{{\min}}} \right) \widehat{\phi}(k)-\tilde{\delta}_{A_{{\min}}} \right) {{ e}^{-Td{k}^{2}}}  \right. \\ \left. + \left( -1+\left( 4-4\widehat{\phi}(k) \right) {{ e}^{T}} \right) \sqrt {{{ e}^{T}}-1	}+ \left( -3+2\widehat{\phi}(k) \right) {{ e}^{T/2}}-4{{ e}^{3/2T}}\left( \widehat{\phi}(k)-1 \right) \right).
	\end{multline*}
	The minimum of the function $\zeta$ is decreasing in $d$ and thus there exists a threshold $d_{A_{\min}}$ such that \eqref{eq: Impulsive: LinStab: Amin case stability condition} does not hold for any $d>d_{A_{\min}}$.
\end{proof}

%%%%%%%%%%%%%%%%%%%%%%%%%%%%%%%%%%%%%%%%%%%%%%%%%%%%%%%%%%%%%%%%%%%%%%%%%%%%%%%%%%%%%%%%%%%%%%%%%%%%%%%%%%%%%%%%%%%%%%%%%%%
%%%%%%%%%%%%%%%%%%%%%%%%%%%%%%             SIMULATIONS             %%%%%%%%%%%%%%%%%%%%%%%%%%%%%%%%%%%%%%%%%%%%%%%%%%%%%%%%
%%%%%%%%%%%%%%%%%%%%%%%%%%%%%%%%%%%%%%%%%%%%%%%%%%%%%%%%%%%%%%%%%%%%%%%%%%%%%%%%%%%%%%%%%%%%%%%%%%%%%%%%%%%%%%%%%%%%%%%%%%% 

\section{Simulations of model extensions}\label{sec: Impulsive: Simulations}

In the preceding linear stability analysis we have made a number of simplifying assumptions to make the derivation of the criteria for pattern onset analytically tractable. To investigate the impact of these simplifications on our results, we numerically investigate extensions of \eqref{eq: Impulsive: intro: model discrete dispersal rain water uptake} in which some previous assumptions are relaxed. 

The analysis in this section yields that the exponential decay (with increasing $T$) of the size of the parameter region supporting pattern onset is due to the temporal separation of the components of the pattern-inducing feedback loop and does not occur if plant growth processes extend into drought periods. Results obtained in this section also highlight the importance of understanding a plant species' response to low soil moisture levels. This functional response is established to have an important influence on the ecosystem dynamics under precipitation regimes with intermediate interpulse times. Finally, the effects of sloped terrain and changes to the plant dispersal kernel are investigated. 

\subsection{Method} \label{sec: Impulsive: Simluations: Methods}
Simulations to determine the parameter region in which pattern onset occurs are performed in two stages. Unless the non-trivial spatially uniform equilibria of the system can be calculated analytically, we initially integrate the corresponding space-independent model to determine the threshold $A_{\min}$ below which the desert equilibrium is the system's only spatially uniform steady state. The calculation of $A_{\min}$ further provides the equilibrium plant and water densities $(\overline{u}^0,\overline{w}^0)$ close to the threshold.

Simulations of the full model are then performed on the space domain $[-x_{\max}, x_{\max}]$ centred at $x=0$. This domain is discretised into $M$ equidistant points $x_1,\dots,x_M$ with $-x_{\max} = x_1 < x_2 < \dots < x_M = x_{\max}$ such that $\Delta x = x_2-x_1 = \dots =x_M - x_{M-1}$. The ODE system resulting from the discretisation of the interpulse PDE system \eqref{eq: Impulsive: intro: model discrete dispersal rain water uptake u in year} and \eqref{eq: Impulsive: intro: model discrete dispersal rain water uptake w in year} is integrated, and the densities at every space point are updated at the end of each interpulse period of length $T$. The discrete convolution term arising from the discretisation of \eqref{eq: Impulsive: intro: model discrete dispersal rain water uptake u between year} is obtained by using the convolution theorem and the fast Fourier transform, providing a significant simplification through a reduction of the number of operations from $O(M^2)$ to $O(M\log(M))$ required to obtain the convolution (see e.g. \cite{Cooley1969}).

To mimic the infinite domain used for the linear stability analysis (Section \ref{sec: Impuslsive: LinStab}), we  define the initial condition of the system as follows; on a subdomain $[-x_{\operatorname{sub}}, x_{\operatorname{sub}} ]$ centred at $x=0$ of the domain $[-x_{\max}, x_{\max}]$ the steady state $(\overline{u}^0, \overline{w}^0)$ near its existence threshold $A_{\min}$ is perturbed by a function containing a collection of applicable spatial modes, while on the rest of the domain the densities are initially set to equal the densities of the steady state $(\overline{u}^0, \overline{w}^0)$. {\color{changes}The restriction of the perturbations to a small subdomain is used to avoid difficulties posed by the boundaries. The size of the outer domain is therefore chosen large enough so that any boundary conditions (which are set to be periodic) that are imposed on $[-x_{\max},x_{\max}]$ do not affect the solution in the subdomain during the time that is considered in the simulation.} Figure \ref{fig: Impuslive: Simulations: example realisation} shows a typical patterned solution obtained by these simulations. 

We use model realisations obtained through this method to determine the critical rainfall level $A_{\max}$ below which pattern onset occurs in the different model extensions.

\begin{figure}
	\centering
	\includegraphics[width=0.9\textwidth]{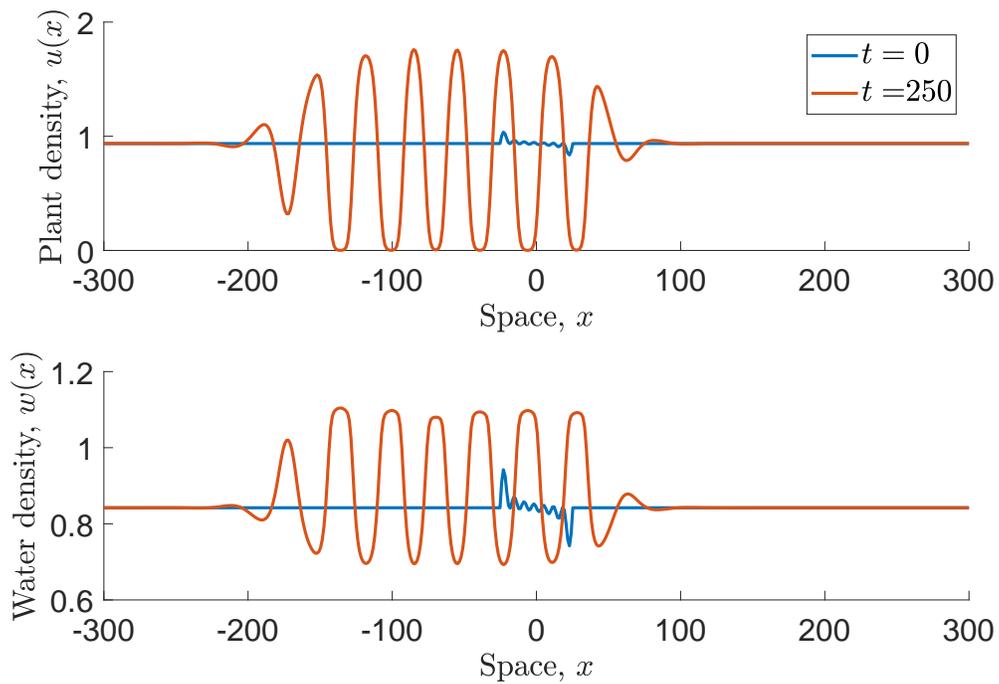}
	\caption{Solution of the impulsive model. This visualises a numerically obtained realisation of the impulsive model \eqref{eq: Impulsive: intro: model discrete dispersal rain water uptake} on flat ground. The plant dispersal kernel $\phi$ is set to the Laplacian kernel \eqref{eq:difference equations: simulations: Laplacian} and the parameter values are $B=0.45$, $A=1.623$, $d=100$ and $T=1$. The number of space points is $M=10^9$.}\label{fig: Impuslive: Simulations: example realisation}
\end{figure}

\subsection{Nonlinear water uptake}\label{sec: Impulsive: Simulations: Nonlinear water uptake}
In the original model \eqref{eq: Impulsive: intro: model discrete dispersal rain water uptake}, water consumption by plants (and the plant growth associated with it) is described by 
\begin{align*}
\operatorname{Up}(u,w) = G_{\operatorname{up}}(w)H_{\operatorname{up}}^2(u) = \left(w+TA\right)\left(\frac{u}{1+u}\right)^2.
\end{align*}
The linearity in $w$ is inherited from the Klausmeier model on which our impulsive model is based. Field observations indicate that dryland ecosystems remain dormant under low soil moisture levels and are only activated if the water density is sufficiently high \cite{Heisler-White2008,Sher2004}. Mathematically, such a property can be described by a Holling type III functional response \cite{Chesson2004}. To incorporate such a nonlinear response into the impulsive model, we consider an amended uptake function with  
\begin{align*}
\widetilde{G_{\operatorname{up}}}(w) = \frac{C_m\left(w+TA\right)^p}{C_h^p + \left(w+TA\right)^p}, \quad p>1,
\end{align*}
where $C_m$ is the maximum water uptake per unit biomass, $C_h$ is the half saturation constant of the water consumption and $p$ accounts for the strength of the nonlinearity. Typical parameter values are $C_m=20$, $C_h=\sqrt{2}$ and $p=4$ \cite{Chesson2004}. The introduction of this nonlinearity causes complications as positivity of the water density $w$ is no longer guaranteed by the update equation \eqref{eq: Impulsive: intro: model discrete dispersal rain water uptake w between year}. To avoid the occurrence of negative densities, we cap the new water uptake function $\widetilde{\operatorname{Up}}(u,w)$ by $w+TA$, i.e. set 
\begin{align*}
\widetilde{\operatorname{Up}}(u,w)=\max\{w+TA,\widetilde{G_{\operatorname{up}}}(w){\color{changes}\}}H_{\operatorname{up}}^2(u).
\end{align*}
 
{\color{changes}The most significant result of our numerical investigation of \eqref{eq: Impulsive: intro: model discrete dispersal rain water uptake} with a Holling type III functional response in the water uptake and plant growth terms is that the minimum of the existence threshold $A_{\min}$ of a non-trivial equilibrium (Figure \ref{fig: Impuslive: Simulations: AT plane nonlinear water uptake}) occurs for intermediate interpulse times.} Under the assumption that total annual rainfall $A$ is fixed, longer drought periods between precipitation pulses correspond to higher intensity rainfall events. Resource availability at the time of water uptake and plant growth is thus higher and exceeds the threshold required for plant growth processes to be activated, which is accounted for in the Holling type III functional response. Conversely, high frequency - low intensity precipitation pulses accumulating to the same amount of total annual rainfall volume are not sufficient to push the water density above this critical value. It is worth emphasising that further increases in the separation of precipitation events (and associated increases in rainfall intensity) to a low frequency - high intensity regime reverses the decrease in $A_{\min}$ due to the saturating behaviour of the water uptake function.

Further, the property that an increase in the interpulse time $T$ reduces the size of the parameter region in which onset of patterns occurs is unaffected by the introduction of a nonlinear water uptake term. Similar to the analytically derived exponential decay of the relative size of $[A_{\min},A_{\max}]$ for \eqref{eq: Impulsive: intro: model discrete dispersal rain water uptake} with a linear functional response (Corollary \ref{cor: Impulsive: LinStab: decrease of pattern supporting interval d infinity limit}), results of our numerical scheme for a Holling type III functional response also indicate an exponential decay of the interval's relative size with increasing interpulse times (Figure \ref{fig: Impuslive: Simulations: pattern forming interval size nonlinear water uptake T}). 

Numerical solutions of the model do, however, become unreliable as the interpulse time $T$ is increased. For larger $T$, the decay-type processes in the interpulse PDEs yield very low plant levels in the troughs of the pattern at the end of the interpulse period. This is a natural source of potential errors. Indeed, Figure \ref{fig: Impuslive: Simulations: numerical issue} depicts that numerical solutions of the system for large $T$ can yield negative plant densities at the end of the interpulse period, highlighting the difficulties encountered in a numerical approach.

To investigate the effects of the strength of the nonlinearity in more detail, we compare results on pattern onset as the strength of the nonlinearity gradually increases away from the linear behaviour considered in \eqref{eq: Impulsive: intro: model discrete dispersal rain water uptake}. While it is impossible to revert back to the linear term by parameter changes only, the behaviour for small values of the water density $w$ can be mimicked by choosing $p=1$ and $C_m=C_h$ sufficiently large. We use this as the reference point to the analytical results obtained in Section \ref{sec: Impuslsive: LinStab} and vary the extent of the nonlinearity in the functional response by fixing $C_m=20$ and setting $C_h = 20-(20-\sqrt{2})\xi$ and $p=1+3\xi$ for $0\le\xi\le1$. 
{\color{changes}For sufficiently low fixed interpulse times $T$, an initial increase of $\xi$ causes an increase of the rainfall level $A_{\min}$ that is required for a spatially uniform non-trivial equilibrium to exist (Figure \ref{fig: Impuslive: Simulations: nonlinear water uptake chi}).} As the strength of the nonlinearity increases further, $A_{\min}$ attains a maximum and then decreases below its level for the model with linear functional water uptake response analysed in Section \ref{sec: Impuslsive: LinStab}. 

The reasoning for this behaviour stems from the variation in the functional response $G_{\operatorname{up}}$ under changes of $\xi$, which is visualised in Figure \ref{fig: Impuslive: Simulations: water uptake functions nonlinear water uptake chi}. For sufficiently low $T$, the resource availability at the time of water uptake is also low. Thus a linear functional response yields a higher water consumption than a nonlinear response with moderate $\xi$, but a lower consumption than a nonlinear response with larger $\xi$. More precisely, the increase in the exponent $p$ and the associated concave-up shape of $G_{\operatorname{up}}$ causes the initial increase in $A_{\min}$. A further increase in $\xi$ decreases the half-saturation parameter $C_h$ and the range of resource densities affected by the concave-up behaviour decreases in size. This causes the eventual decrease in $A_{\min}$ as the strength of the nonlinearity is increased further. 

{\color{changes}For sufficiently large drought lengths $T$, the maximum in $A_{\min}$ occurs at $\xi=0$ and thus any $\xi>0$ reduces the minimum water requirements of the system.} The upper bound $A_{\max}$ of the parameter region supporting pattern onset mimics the behaviour of $A_{\min}$. The size of the parameter region in which pattern onset occurs increases slightly with increasing $\xi$, but changes to its size are insignificant compared to changes causes by variations in the interpulse time $T$.

\begin{figure}
	\centering
	\subfloat[Classification of the $A$-$T$ parameter plane. \label{fig: Impuslive: Simulations: AT plane nonlinear water uptake}]{\includegraphics[width=0.32\textwidth]{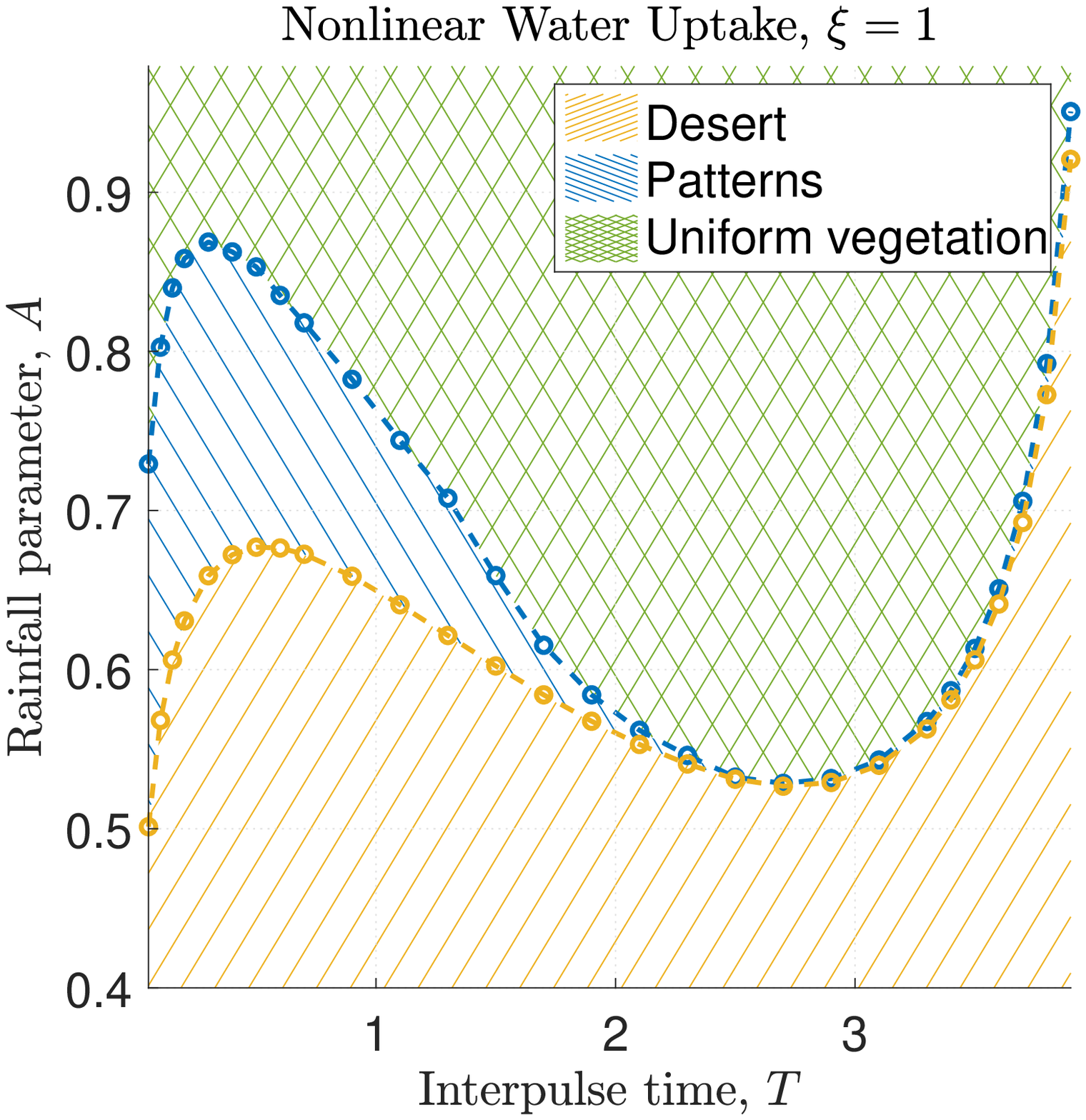}} \hfill
	\subfloat[Relative size of the rainfall interval supporting pattern onset.\label{fig: Impuslive: Simulations: pattern forming interval size nonlinear water uptake T}]{\includegraphics[width=0.32\textwidth]{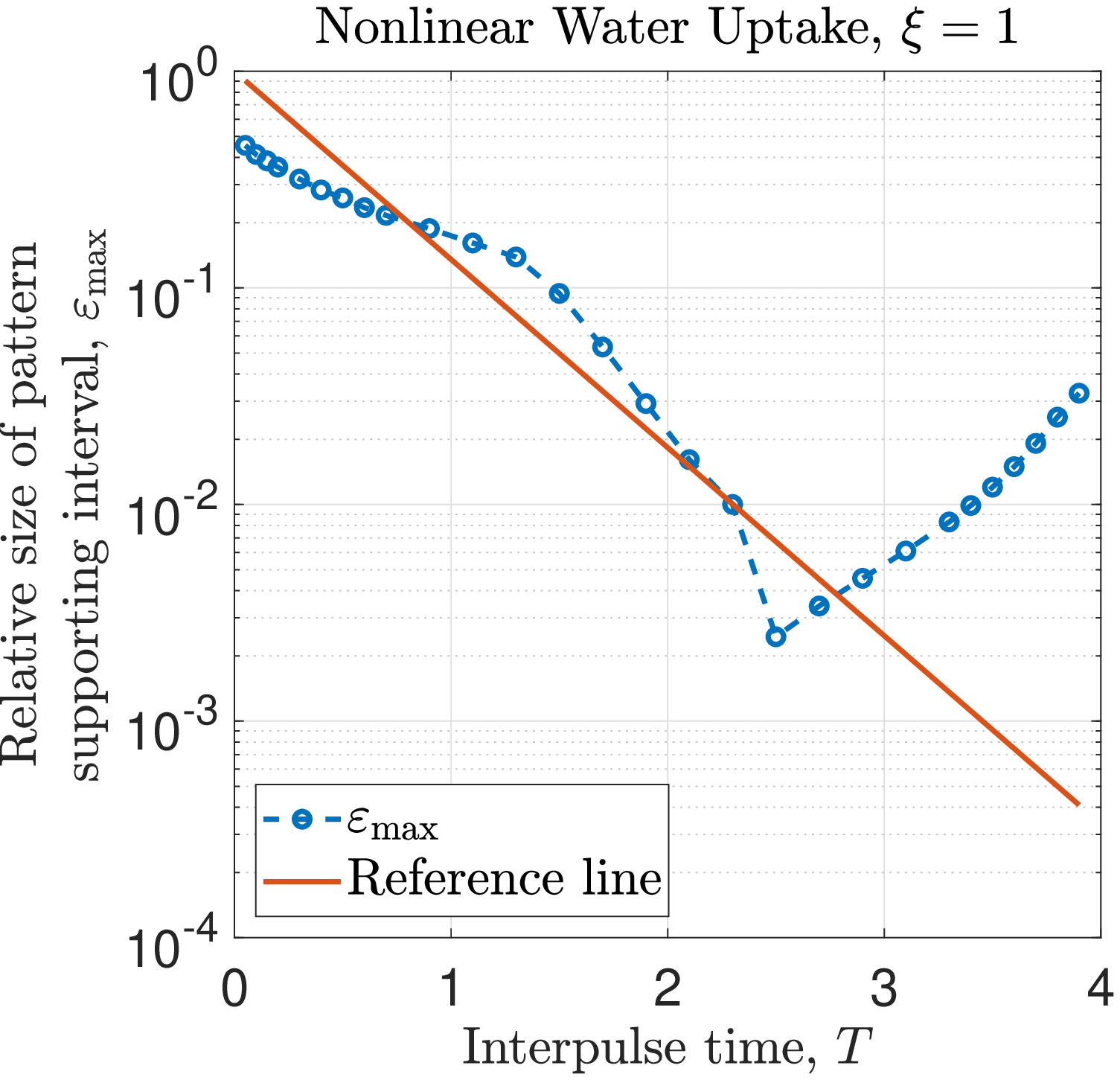}} \hfill
	\subfloat[Minimum plant density\label{fig: Impuslive: Simulations: numerical issue}]{\includegraphics[width=0.32\textwidth]{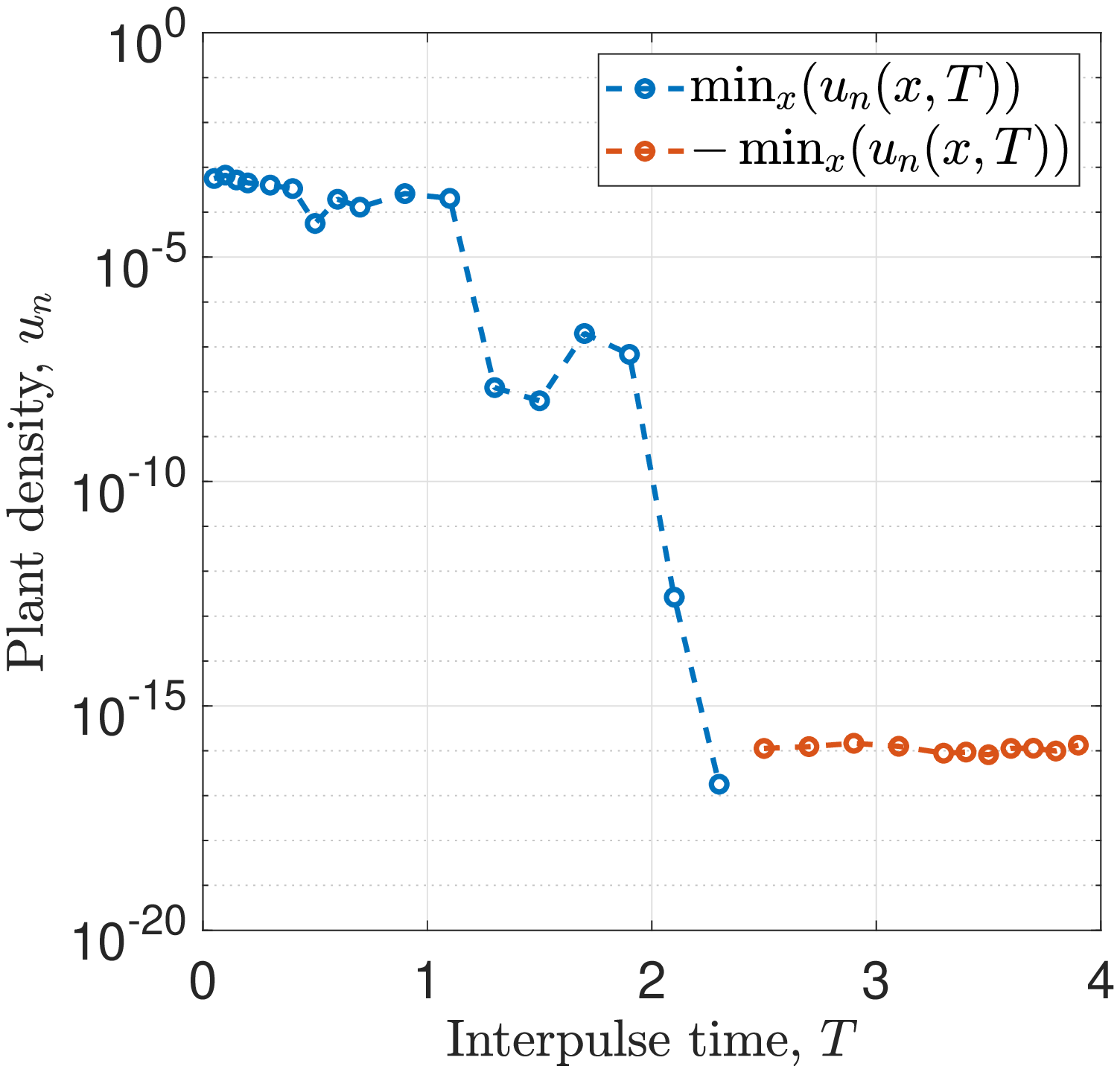}} \\
	\subfloat[Classification of the $A$-$\xi$ parameter plane. \label{fig: Impuslive: Simulations: nonlinear water uptake chi}]{\includegraphics[width=0.32\textwidth]{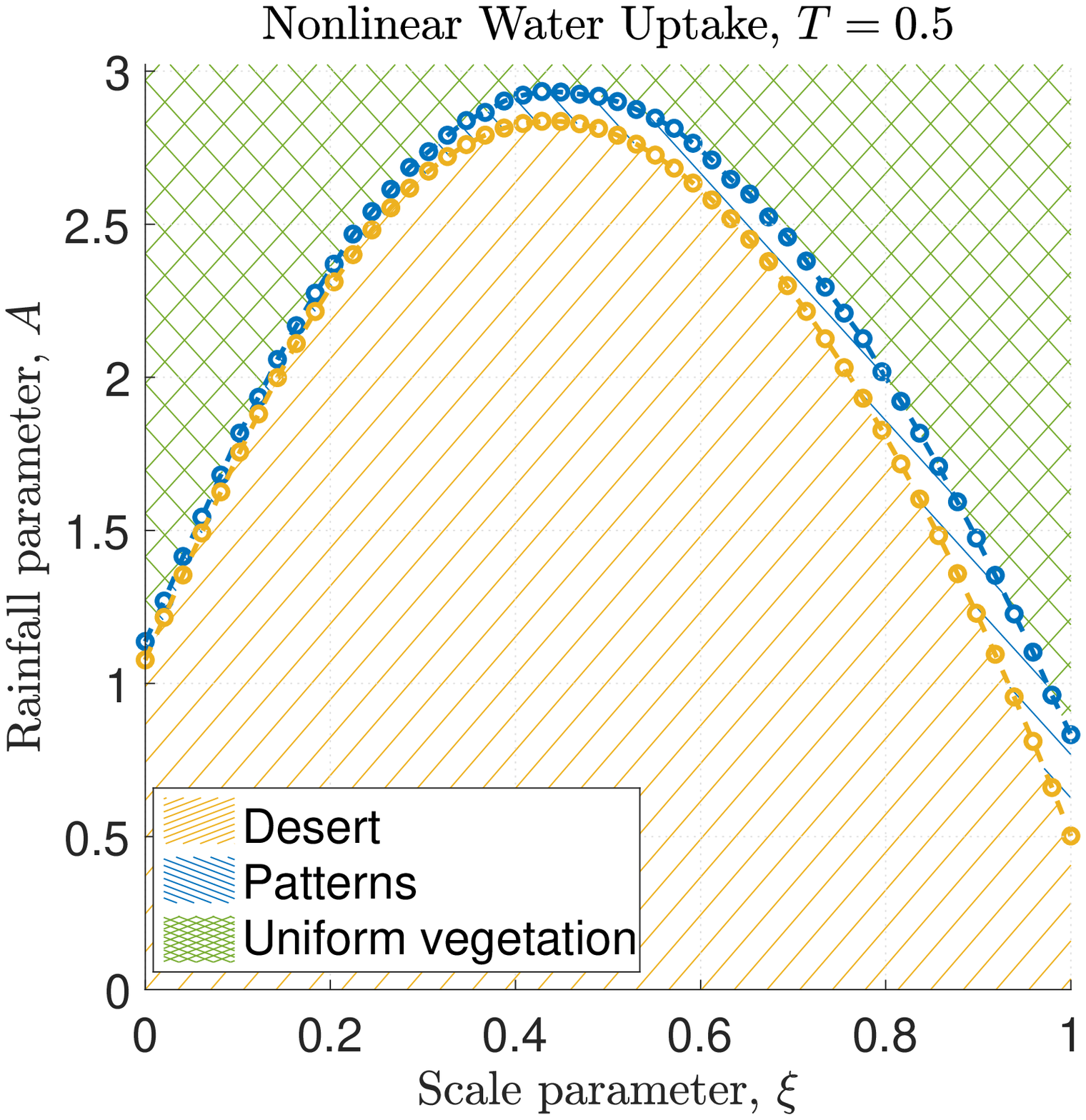}} \hspace{1cm}
		\subfloat[Water uptake functional response. \label{fig: Impuslive: Simulations: water uptake functions nonlinear water uptake chi}]{\includegraphics[width=0.32\textwidth]{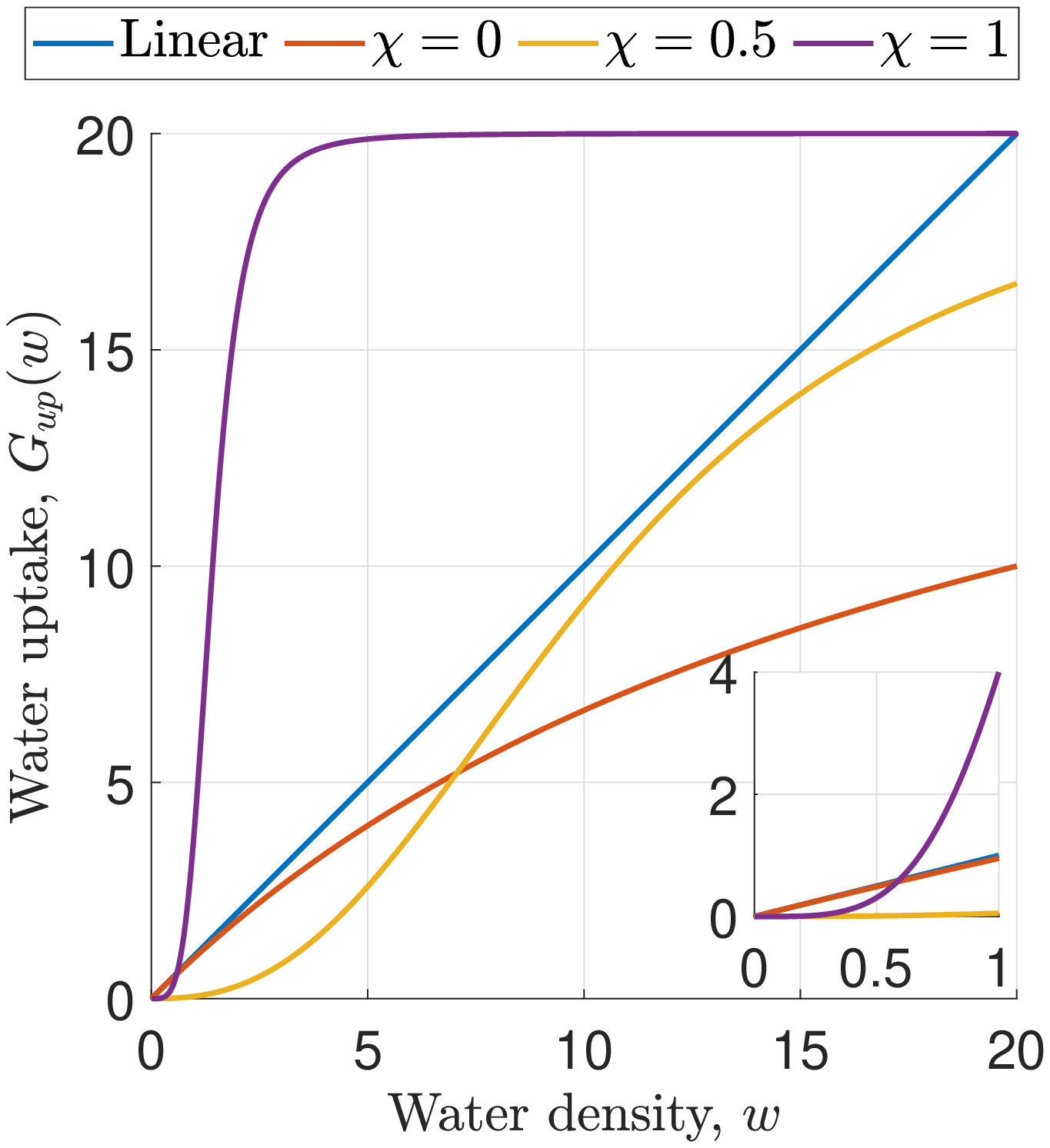}}
		\caption{Classification plots for a nonlinear functional response in the water uptake function. The classifications (a) and (d) into states of desert, onset of spatial patterns and uniform vegetation are based on the numerical scheme described in Section \ref{sec: Impulsive: Simluations: Methods}. The transition threshold $A_{\max}$ is determined up to an interval of size $10^{-5}$, the level of $A_{\min}$ up to an interval of size $10^{-8}$. The relative size of $[A_{\min},A_{\max}]$ corresponding to the classification in (a) is shown in (b), where the reference line is of slope $\exp(-2T)$. The parameter values used in both simulations are $B=0.45$ and $d=500$. The water uptake function $G_{\operatorname{up}}(w)$ is shown in (e) for several values of $\xi$, with its behaviour close to the origin shown in the inset. The minimum plant density before a rainfall pulse $\min_{x\in[-x_{\max},x_{\max}]}\{u_n(x,T)\}$ of a stable pattern is shown in (c), where the blue and red markers indicate positive and negative values of $u_n$, respectively. This visualises the numerical issues encountered in simulations for longer interpulse times $T$. }\label{fig: Impuslive: Simulations: nonlinear water uptake T}
\end{figure}

\subsection{Nonlinear PDEs} \label{sec: Impulsive: Simulations: Nonlinear PDEs}
The original impulsive model \eqref{eq: Impulsive: intro: model discrete dispersal rain water uptake} is based on the assumption that no plant-water interactions take place during drought periods. The interpulse equations thus form a system of linear and decoupled PDEs that describe linear decay of both plant and water densities between precipitation pulses. We relax this assumption by extending the occurrence of biomass growth into the interpulse phase. This changes the PDE system to
\begin{align*}
\frac{\partial u_n}{\partial t} &= -Bu_n + C\left(\frac{u_n}{1+u_n}\right)^2w_n, \\
\frac{\partial w_n}{\partial t} &= -w_n - C\left(\frac{u_n}{1+u_n}\right)^2w_n + d \frac{\partial^2 w_n}{\partial x^2}, 
\end{align*} 
where the nondimensional constant $C$ accounts for the rate of water uptake. The pulse equations \eqref{eq: Impulsive: intro: model discrete dispersal rain water uptake u between year} and \eqref{eq: Impulsive: intro: model discrete dispersal rain water uptake w between year} remain unchanged, i.e. there is still a pulse of plant growth synchronised with a precipitation event.

While a typical estimate is $C=10$ \cite{Gilad2007}, we use our numerical scheme to investigate how a gradual increase from $C=0$ (which corresponds to the model studied analytically in Section \ref{sec: Impulsive}) affects the pattern onset observed in the system. {\color{changes}An increase in the plants' growth rate during drought periods causes a decrease in the existence threshold $A_{\min}$ of a spatially uniform non-desert equilibrium and the precipitation level $A_{\max}$ at which pattern onset occurs (Figure \ref{fig: Impuslive: Simulations: Achi plane nonlinear water uptake}).} This decrease is caused by a reduction in total resource loss through evaporation and the associated increase in water availability to plants. In the original model \eqref{eq: Impulsive: intro: model discrete dispersal rain water uptake} ($C=0$), water that is not consumed by plants during the rainfall pulse undergoes exponential decay due to evaporation during the interpulse period and is lost from the system. If $C\ne 0$, however, water that enters the drought phase not only evaporates but also continues to be converted into plant biomass, which causes a reduction in evaporation losses.

The second main conclusion arising from the inclusion of a nonlinear coupling of the interpulse PDEs is the conservation of a large parameter region in which pattern onset occurs for large $T$ (Figure \ref{fig: Impuslive: Simulations: nonlinear PDEs T}), instead of an exponential decay of its size with increasing $T$. The existence of such a region is due to the inclusion of a pattern-inducing feedback in the interpulse PDEs. More water is consumed in regions of high biomass density, which causes the homogenising effect of water diffusion to redistribute more water towards these regions yielding further plant growth. If water uptake between pulses is weak (small $C$), or as in the original model non-existent ($C=0$), the system's only pattern-forming feedback loop consists of the nonlinearity in the plant growth term in the update equations in combination with the homogenising property of water diffusion in the interpulse PDEs. The latter loses its impact as $T$ is increased, as evaporation effects become dominant and cause a decrease in water availability at the end of the interpulse phase. The water density at the growth pulse therefore only depends on the intensity of the rain event, but is independent of the diffusion process that occurs before the rainfall pulse. This weakens the strength of the pattern-inducing feedback loop and causes the decrease in the size of the parameter region in which pattern onset occurs.

\begin{figure}
	\centering
	\subfloat[Classification of the $A$-$T$ parameter plane.\label{fig: Impuslive: Simulations: nonlinear PDEs T}]{\includegraphics[width=0.48\textwidth]{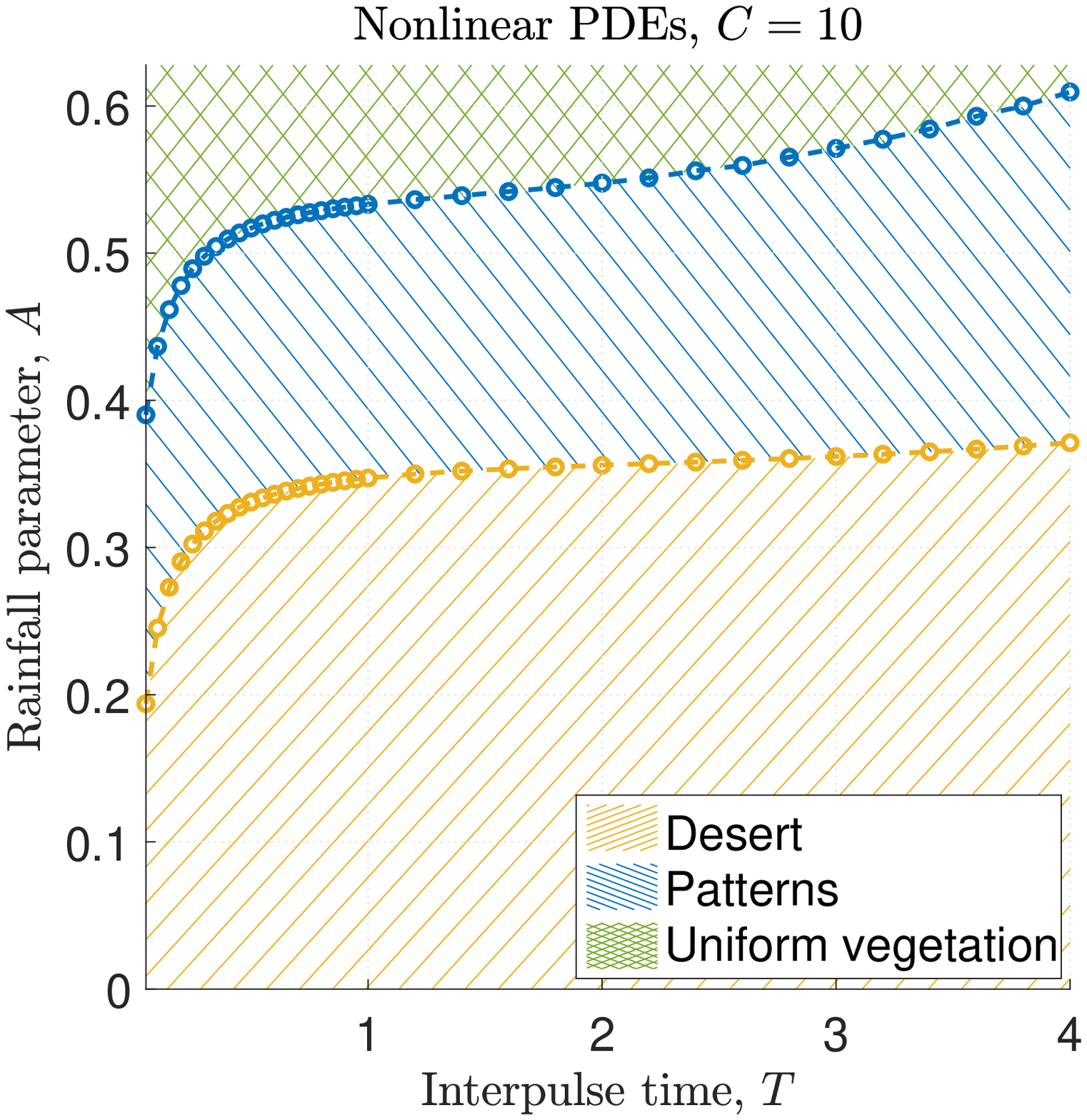}}
	\subfloat[Classification of the $A$-$C$ parameter plane. \label{fig: Impuslive: Simulations: Achi plane nonlinear water uptake}]{\includegraphics[width=0.48\textwidth]{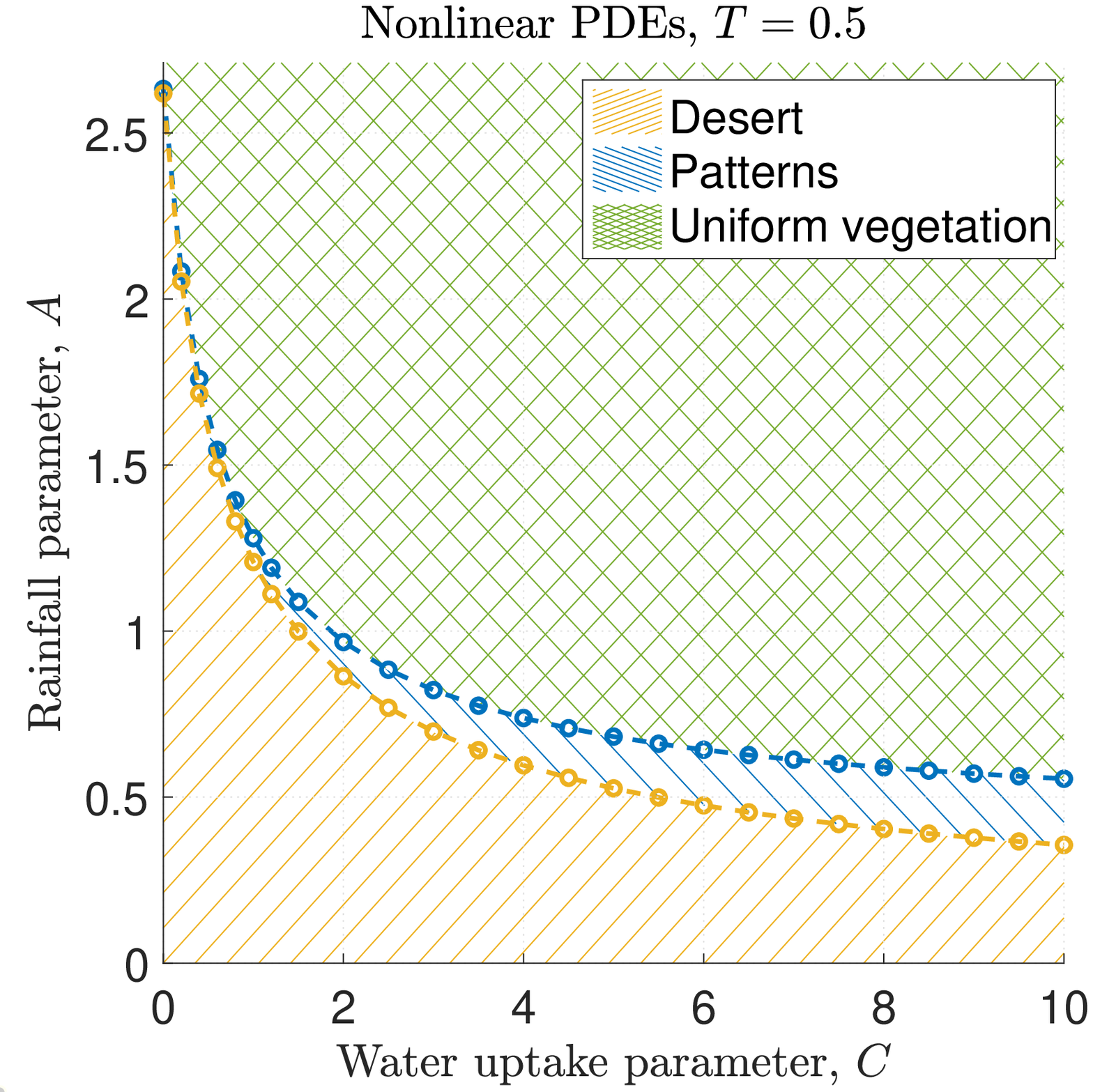}}
	\caption{Classification plots for the inclusion of plant growth in the interpulse PDEs. The classifications into states of desert, onset of spatial patterns and uniform vegetation is based on the numerical scheme described in Section \ref{sec: Impulsive: Simluations: Methods}. The transition threshold $A_{\max}$ is determined up to an interval of size $10^{-5}$, the level of $A_{\min}$ up to an interval of size $10^{-8}$. The parameter values used in both simulations are $B=0.45$ and $d=500$.}\label{fig: Impuslive: Simulations: nonlinear PDEs chi}
\end{figure}

\subsection{Kernel functions}
In the linear stability analysis in Section \ref{sec: Impuslsive: LinStab}, we set the plant dispersal kernel to the Laplace kernel \eqref{eq:difference equations: simulations: Laplacian}. Seed dispersal behaviour, however, depends both on species and environmental conditions \cite{Bullock2017}. Similar to the work on a previous model \cite{Eigentler2018nonlocalKlausmeier}, we use our numerical scheme to investigate effects caused by setting the dispersal kernel to the Gaussian
\begin{align}\label{eq:difference equations: simulations: Gaussian}
\phi(x) = \frac{a_g}{\sqrt{\pi}}e^{-a_g^2x^2}, \quad a>0, x\in\R,
\end{align}
and the power law distribution
\begin{align}\label{eq:difference equations: simulations: Power Law}
\phi(x) = \frac{(b-1)a_p}{2\left(1+a_p|x|\right)^b}, \quad a>0, b >3, x\in\R.
\end{align}

We base our comparison on the kernels' standard deviations, which are given by $\sigma(a) = \sqrt{2}/a$ for the Laplacian kernel \eqref{eq:difference equations: simulations: Laplacian}, $\sigma(a_g) = 1/(\sqrt{2}\,a_g)$ for the Gaussian kernel \eqref{eq:difference equations: simulations: Gaussian} and $\sigma(a_p) = \sqrt{2}/(\sqrt{b^2-5b+6}\,a_p)$ for the power law kernel \eqref{eq:difference equations: simulations: Power Law} provided $b>3$. It is perfectly reasonable to perform simulations with kernels of infinite standard deviation (e.g. $b<3$ in the power law kernel) but in the interest of comparing results for the kernels based on their standard deviation we consider only $b=3.1$ and $b=4$.

In the simulations we are interested in both the effects of changes to the shape of the dispersal kernel and the effects caused by a variation in the temporal intermittency of precipitation. As shown in Figure \ref{fig: Impuslive: Simulations: Amax changes T flat}, the latter bears much more influence on the rainfall threshold $A_{\max}$ than the choice of plant dispersal kernel. Indeed, the results obtained for all kernel functions follow the narrow band of exponentially decaying size in the $T$-$A$ parameter region in which pattern existence has been shown for the Laplace kernel in Section \ref{sec: Impuslsive: LinStab} and in particular in Figure \ref{fig: Impuslive: LinStab: dlimit AT plane full}.

While the effects of the kernel shape are negligible compared to changes of the interpulse time $T$, their influence on the system can still be studied if $T$ is fixed. Instead of varying $T$, we opt to investigate how the threshold $A_{\max}$, at which patterns cease to exist, changes under variations of the water diffusion coefficient $d$. This allows us to draw a connection to the results of the linear stability analysis visualised in Figure \ref{fig: Impulsive: LinStab: jury2 in Ad parameter plane flat ground}. Our numerical scheme shows that all kernel functions considered in the simulations qualitatively follow the same behaviour, which agrees with the analytically deduced result for the Laplace kernel in Section \ref{sec: Impuslsive: LinStab}. For sufficiently low levels of rainfall, the diffusion coefficient needs to exceed a threshold to give rise to an instability resulting in the onset of spatial patterns. There does, however, exist an upper bound (not shown in Figure \ref{fig: Impuslive: Simulations: Amax changes diffusion}) on the rainfall parameter for each kernel function above which pattern onset from a perturbation of the steady state is impossible. Due to the nondimensionalisation of the model an increase in the diffusion coefficient $d$ corresponds to a decrease in the width of the dispersal kernels. Thus, for a fixed kernel function an increase in kernel width inhibits the onset of patterns. Note, however, that information on the kernels' standard deviation, which we use as a measurement of kernel width, is insufficient to make comparisons between results for different kernel functions. Conditions for pattern onset also depend on the dispersal kernel's type of decay at infinity; for example $A_{\max}$ for the Laplace kernel and the power law kernel with $b=4$ coincide in Figure \ref{fig: Impuslive: Simulations: Amax changes diffusion}, even though their standard deviations are $\sigma_L = \sqrt{2}$ and $\sigma_P = 1$, respectively.

\begin{figure}
	\centering
	\subfloat[\label{fig: Impuslive: Simulations: Amax changes diffusion}]{\includegraphics[width=0.48\textwidth]{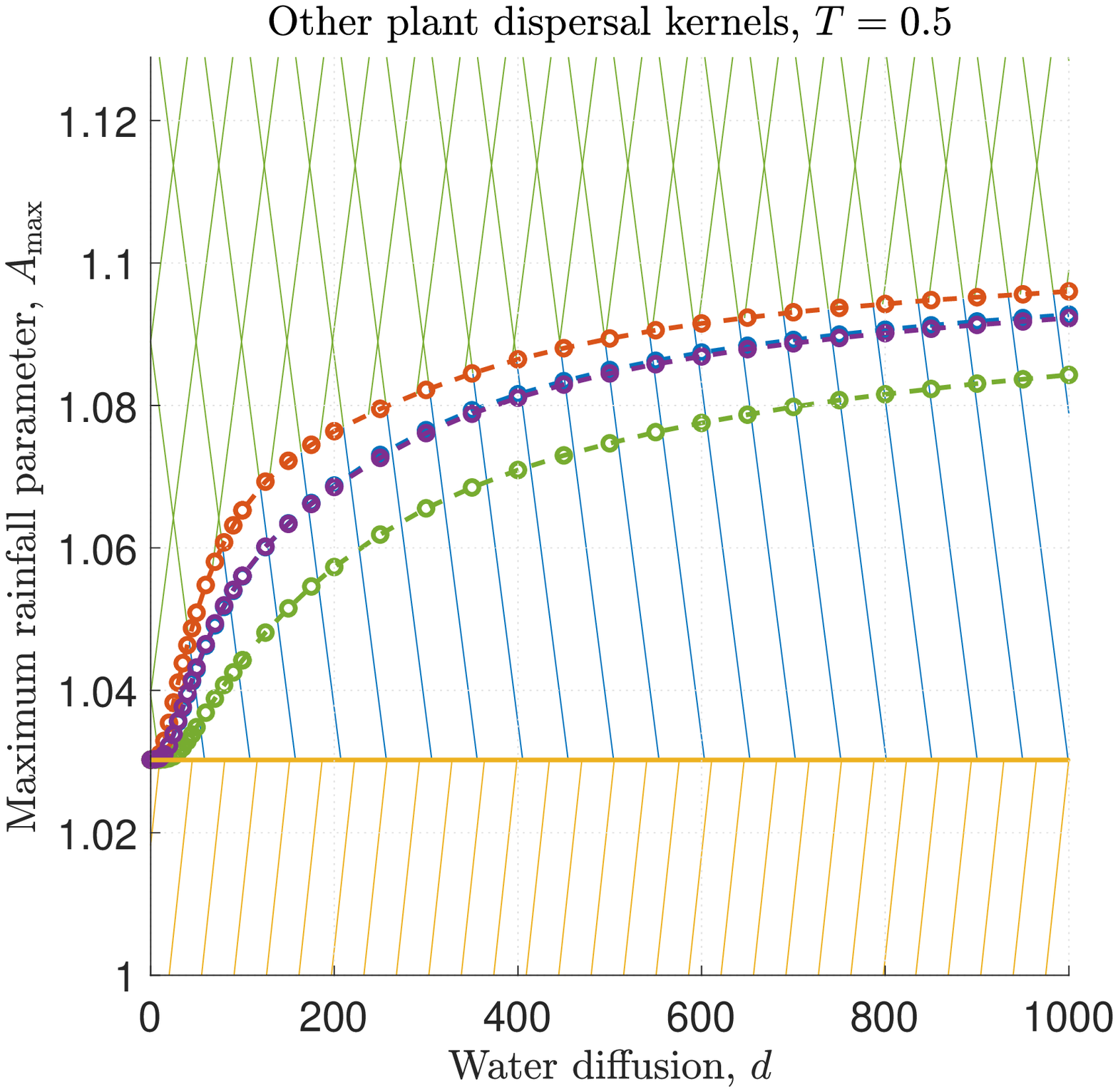}}
	\subfloat[\label{fig: Impuslive: Simulations: Amax changes T flat}]{\includegraphics[width=0.48\textwidth]{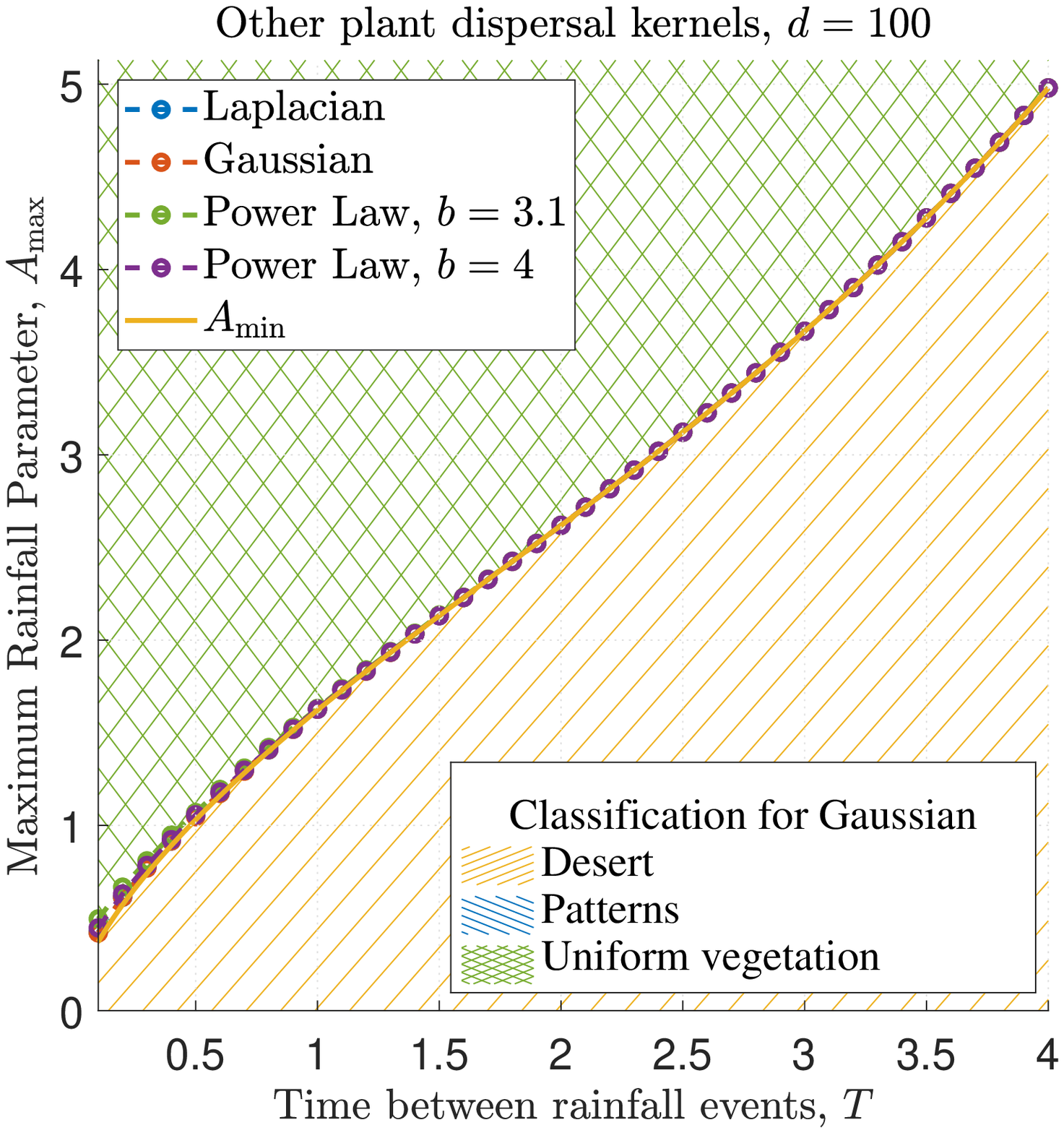}}
	\caption{Changes to $A_{\max}$ under variation of water diffusion and the time between rain pulses. This figure visualises changes to the critical rainfall parameter $A_{\max}$ under changes of the water diffusion rate $d$ ((a)) and the interpulse time $T$ ((b)). The rainfall threshold $A_{\max}$ is determined up to an interval of length $10^{-4}$ for $d=\{0,5,  \dots, 50, 60 \dots 100, 125, \dots, 200, 250, \dots, 1000 \}$ and $T=\{0.1,0.2,\dots, 4\}$, respectively. Plant mortality is set to $B=0.45$. The legend applies to both parts of the figure.}	
\end{figure}

\subsection{Slope}\label{sec: Impulsive: Simulations: Slope}

Finally, we lift the restriction of the flat spatial domain for which the linear stability analysis of \eqref{eq: Impulsive: intro: model discrete dispersal rain water uptake} was performed in Section \ref{sec: Impulsive}. Originally, the Klausmeier model was proposed to describe vegetation bands on sloped terrain and a lot of previous work has focussed on this scenario (e.g. \cite{Klausmeier1999, Sherratt2013IV, Eigentler2018nonlocalKlausmeier}). 
A numerical investigation into the existence of spatial patterns of \eqref{eq: Impulsive: intro: model discrete dispersal rain water uptake} on a sloped spatial domain shows that the threshold $A_{\max}$ at which a transition between uniform and patterned vegetation occurs, increases with increasing slope $\nu$ (Figure \ref{fig: Impuslive: Simulations: slope chi}). The lower bound $A_{\min}$ of the parameter region supporting the onset of spatial patterns from spatially nonuniform perturbations of the equilibrium, is a non-spatial property and thus independent of the slope parameter $\nu$. Thus the size of $[A_{\min},A_{\max}]$ increases with increasing $\nu$. Ecologically, this stems from an increase in the strength of the pattern-forming mechanism. On steeper slopes water flows downhill faster and thus increases the competitive advantage of existing biomass patches.

This increase in the size of $[A_{\min},A_{\max}]$ is, however, negligible compared to the decay of the interval's size for increasing interpulse time $T$ (Figures \ref{fig: Impuslive: Simulations: AT plane slope} and \ref{fig: Impuslive: Simulations: pattern forming interval size slope T}). Our results indicate that the interval's size decays exponentially, similar to the analytically obtained result (Corollary \ref{cor: Impulsive: LinStab: decrease of pattern supporting interval d infinity limit}) for the model on flat ground in the limit $d\rightarrow \infty$. We thus conclude that the simplified model ($\nu=0$) qualitatively yields the same results on the onset of patterns under variations in the length of the drought period $T$.

\begin{figure}
	\centering
	\subfloat[Classification of the $A$-$T$ parameter plane. \label{fig: Impuslive: Simulations: AT plane slope}]{\includegraphics[width=0.48\textwidth]{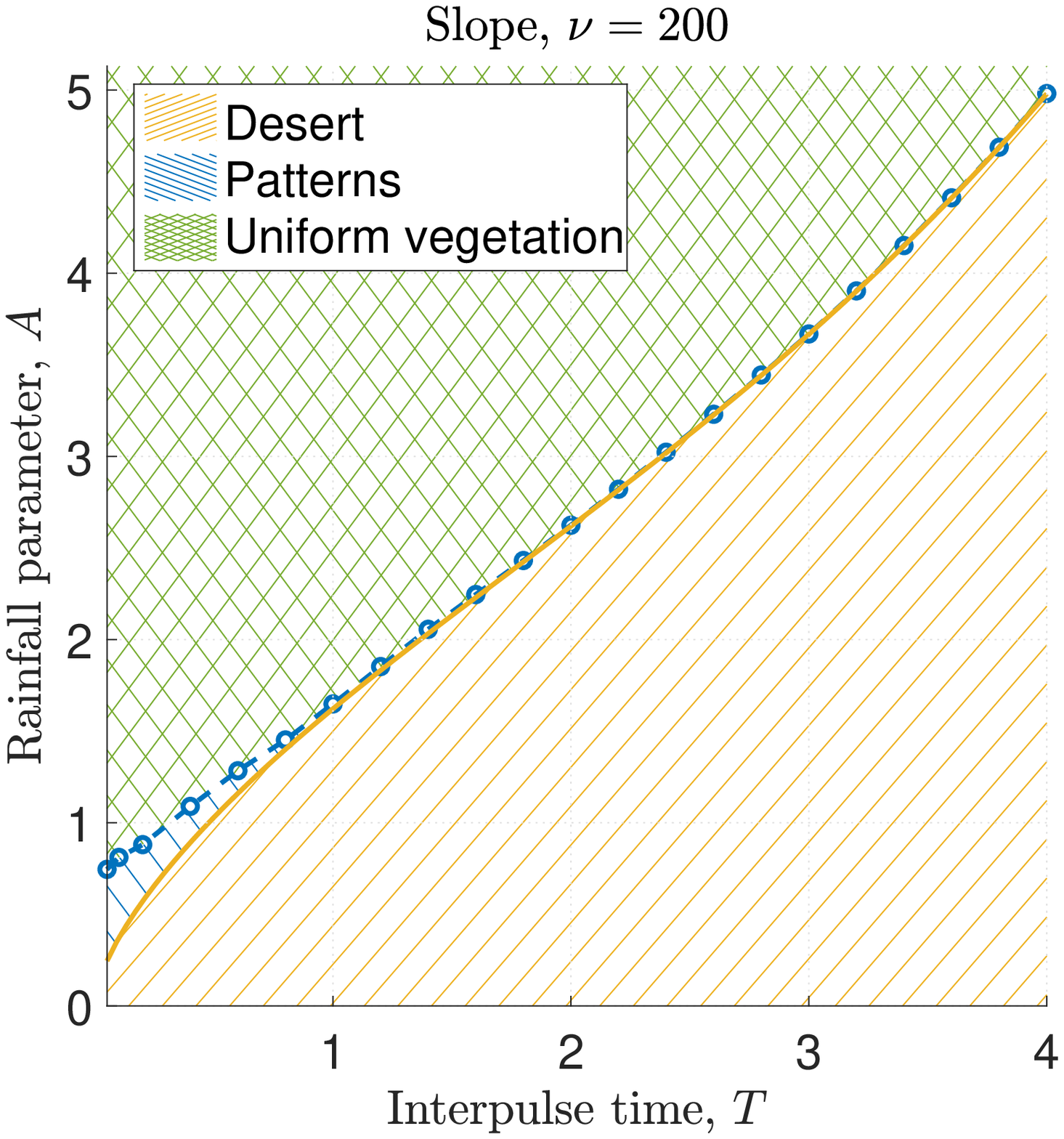}}
	\subfloat[Relative size of the rainfall interval supporting pattern onset.\label{fig: Impuslive: Simulations: pattern forming interval size slope T}]{\includegraphics[width=0.48\textwidth]{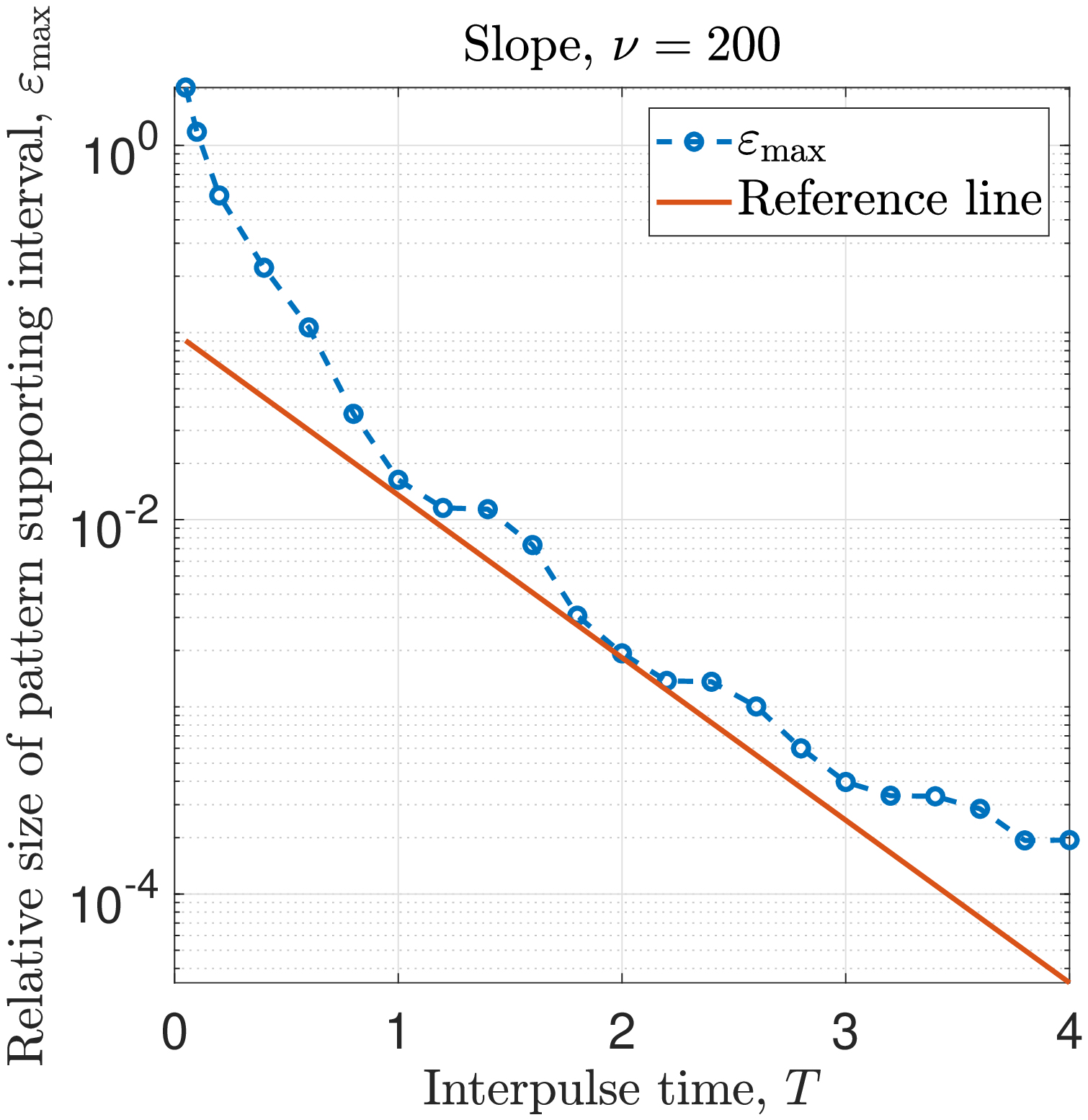}} \\
	\subfloat[Classification of the $A$-$\nu$ parameter plane. \label{fig: Impuslive: Simulations: slope chi}]{\includegraphics[width=0.48\textwidth]{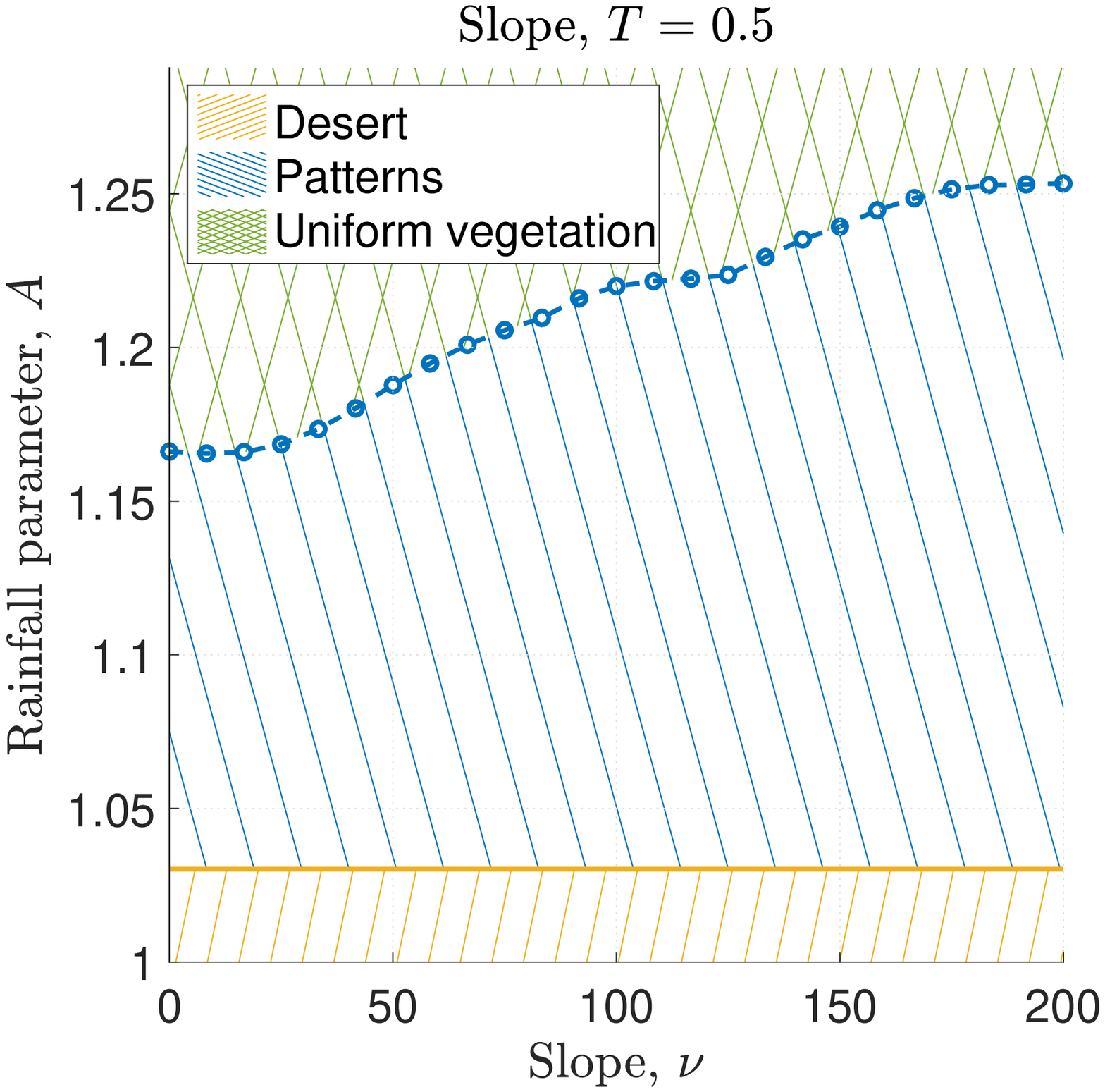}}
	\caption{Classification plots for the model on a slope. The classifications into states of desert, onset of spatial patterns and uniform vegetation are based on the numerical scheme described in Section \ref{sec: Impulsive: Simluations: Methods}. The transition threshold $A_{\max}$ is determined up to an interval of size $10^{-8}$, the level of $A_{\min}$ is given by \eqref{eq: Impulsive: LinStab: lower bound on A for ss to exist}. The relative size of $[A_{\min}, A_{\max}]$ is shown in (b), where the reference line has slope $\exp(-2T)$. The parameter values used in both simulations are $B=0.45$ and $d=500$.}\label{fig: Impuslive: Simulations: slope}
\end{figure}

%%%%%%%%%%%%%%%%%%%%%%%%%%%%%%%%%%%%%%%%%%%%%%%%%%%%%%%%%%%%%%%%%%%%%%%%%%%%%%%%%%%%%%%%%%%%%%%%%%%%%%%%%%%%%%%%%%%%%%%%%%%
%%%%%%%%%%%%%%%%%%%%%%%%%%%%%%             DISCUSSION             %%%%%%%%%%%%%%%%%%%%%%%%%%%%%%%%%%%%%%%%%%%%%%%%%%%%%%%%
%%%%%%%%%%%%%%%%%%%%%%%%%%%%%%%%%%%%%%%%%%%%%%%%%%%%%%%%%%%%%%%%%%%%%%%%%%%%%%%%%%%%%%%%%%%%%%%%%%%%%%%%%%%%%%%%%%%%%%%%%%% 

\section{Discussion}\label{sec: Discussion}

In this paper we consider a new impulsive-type model to investigate the effects of rainfall intermittency on the onset of vegetation patterns in semi-arid environments. Most significantly, our results suggest that the decay-type behaviour which dominates during long drought periods inhibits the onset of spatial patterns and that ecosystems benefit from precipitation intermittency if plant species are unable to efficiently use low soil moisture levels.

The inhibition of patterns by low frequency rain events is quantified by the small size of the interval of rainfall levels in which pattern onset occurs. Therefore, plants are able to form a uniform vegetation cover for rainfall levels very close to the minimum required for the corresponding spatially uniform equilibrium to exist. This pattern-inhibitory effect in the impulsive model (see Proposition \ref{cor: Impulsive: LinStab: decrease of pattern supporting interval d infinity limit} and Figures \ref{fig: Impuslive: LinStab: dlimit AT plane} and \ref{fig: Impuslive: Simulations: Amax changes T flat}) can be explained by the weakening effect of the temporal separation of rainfall pulses on the plant growth-water redistribution feedback which is the main contributor in the formation of patterns \cite{Rietkerk2008, Lejeune1999}. This positive feedback loop consists of two processes; the increased water utilisation in regions of high biomass and the redistribution of water. In \eqref{eq: Impulsive: intro: model discrete dispersal rain water uptake} these processes occur in different stages. The soil modification by plants affects water consumption and plant growth which only occur in the update equations associated with a rainfall pulse, while water diffusion is accounted for in the interpulse PDE system. Therefore, if plants are in a patterned state, the water density immediately after a rainfall pulse is in antiphase to the plant density (i.e. high water density in regions of low biomass and vice versa). The homogenising property of diffusion thus redistributes water from patches of low biomass to regions where plant density is high. If, however, the separation of precipitation pulses is too long, this homogenising effect loses its impact as water evaporation becomes the dominant process.
In the model extension which assumes that plant growth occurs in both the pulse stage and during the interpulse period (Section \ref{sec: Impulsive: Simulations: Nonlinear PDEs}), the temporal separation of rainfall events does not weaken the pattern-inducing feedback. The closure of the feedback loop in the interpulse PDEs allows for more water transported to regions of high biomass during drought periods to be utilised and thus supports the pattern-forming mechanism. 

For a fixed interpulse time $T$, the reduction in water evaporation associated with this increase in water to biomass conversion causes a reduction in the minimum amount of precipitation required for a spatially uniform equilibrium to exist. We use this minimum on the rainfall parameter ($A_{\min}$) as a proxy for the minimum water requirements of the ecosystem, but emphasise the fact that spatially non-uniform stable states with non-zero plant densities are likely to exist for lower precipitation levels and no information on the resilience of the ecosystem can be extracted from the analysis presented in this paper. For both the extension with nonlinear interpulse PDEs and the original model \eqref{eq: Impulsive: intro: model discrete dispersal rain water uptake}, the threshold $A_{\min}$ increases with the drought period length $T$, which indicates that an increase in the time between rainfall events has a detrimental effect on the ecosystem. Even though this does not agree with the majority of reported field observations \cite{Sher2004, Heisler-White2008}, there exists evidence of this inhibitory effect for some dryland species, with an increase in seeds germination rates, a decrease in emergence rates and an increase in seedling mortalities under longer periods of droughts \cite{Lewandrowski2017, Lundholm2004}. This suggests that an ecosystem's response to temporal variability in precipitation is highly species-dependent and it is important to understand a species' response to oscillations in soil moisture to model its dynamics. Indeed, we have established that changes to the plants' water uptake functional response to the water density (Section \ref{sec: Impulsive: Simulations: Nonlinear water uptake}) can reverse the increasing behaviour of the minimum water requirement proxy $A_{\min}$ observed in the original model in which the functional response is linear (Figure \ref{fig: Impuslive: LinStab: dlimit AT plane full}). If species in an ecosystem remain dormant under low soil moisture levels caused by a high frequency - low intensity rainfall regime, then rainfall intermittency and the associated temporal increases in soil moisture can have a positive impact on the ecosystem \cite{Heisler-White2008, Sher2004}. Mathematically, we used a Holling type III functional response to model this dormant behaviour under low soil moisture levels. If the concave-up shape of this species-dependent functional response is sufficiently strong for low water densities, then $A_{\min}$ attains a minimum for an intermediate interpulse time $T$ because water uptake is maximised under such conditions. This is in agreement with results obtained for the Gilad \cite{Gilad2004} model \cite{Kletter2009}.

The dominant role of precipitation intermittency on the onset of patterns also manifests itself in the fact that, unlike in the Klausmeier models, diffusion alone is insufficient to cause pattern onset in the impulsive model. The onset of spatial patterns still requires the diffusion coefficient to exceed a threshold (Proposition \ref{prop: Impulsive: LinStab: dc in A=Amin case flat ground}) but in stark contrast to the Klausmeier models in which a sufficiently large level of diffusion can cause an instability for an arbitrarily large level of rainfall, the effects of diffusion are limited to a small interval of the rainfall parameter (Proposition \ref{prop: Impulsive: LinStab: Amax in d infinity limit}), whose size decreases exponentially as precipitation pulses become more infrequent (Corollary \ref{cor: Impulsive: LinStab: decrease of pattern supporting interval d infinity limit}). This deviation form the classical case of a diffusion-driven instability is due to the previously discussed temporal separation of the components of the pattern-inducing feedback that renders diffusion effects insignificant under long drought spells. This property is specific to the system considered in this paper and no generalisations can be made. Indeed, diffusion-driven instabilities have been shown to occur in other impulsive models \cite{Wang2018}.

A second key aspect of this paper is the effects caused by changes to the width and shape of the plant dispersal kernel. Contrary to the beneficial effect {\color{changes}associated with the inhibition of pattern onset due to} wide plant dispersal kernels shown by the model in this paper (Figure \ref{fig: Impuslive: Simulations: Amax changes diffusion}), plants in semi-arid regions are observed to establish narrow dispersal kernels \cite{Ellner1981}. This is, however, only a secondary effect caused by other adaptations such as protection from seed predators, that are not accounted for in these models but nevertheless affect the vegetation's evolution in arid regions \cite{Ellner1981}. The quantitatively small changes to the rainfall threshold $A_{\max}$ in the impulsive model are caused by the fact that in the impulsive model only the newly added biomass is dispersed, while in the other models the whole plant density undergoes dispersal. Combined with the claim that plants compensate for the negative effect of a narrow seed dispersal kernel by changes of traits not included in this model, this suggests the combination of the weak response of the impulsive model to changes in the width of the dispersal kernel and the stronger effect of rainfall intermittency provides a more realistic framework than a previous model in which the seed dispersal distance played an important role in the absence of any pulse-type events \cite{Eigentler2018nonlocalKlausmeier}.

To facilitate the mathematical analysis presented in this paper we have opted for a fully deterministic modelling of precipitation. The assumption that rainfall events occur periodically in time and are all of the same intensity is, however, an inaccurate description of the inherently stochastic nature of this key process. A more realistic description of such precipitation events can be achieved through a Poisson process with exponentially distributed rainfall intensities \cite{Rodriguez-Iturbe1999}. The model framework presented in this model is, however, insufficient to consider any stochasticity in the rainfall regime. Neither the original model \eqref{eq: Impulsive: intro: model discrete dispersal rain water uptake} nor any of its extensions presented in Section \ref{sec: Impulsive: Simulations} include mechanisms that allow plants to recover from a very low density. Thus, the eventual occurrence of a long drought period (possibly combined with low intensity pulses) under a stochastic precipitation regime inevitably causes the extinction of plants in the long term. In reality, plants have developed mechanisms such as seed dormancy that allow recovery from low biomass densities \cite{Lewandrowski2017}. Their inclusion in a mathematical model is required to better understand an ecosystem's response to stochasticity in environmental conditions. Nevertheless, it is possible to relate the results of the deterministic model presented in this paper to a stochastic setting. Similar to a previous study of effects of temporal variations of rainfall pulses on dryland ecosystems, the constants involved in the deterministic modelling of precipitation can be seen as the expected values that arise from the underlying stochastic processes \cite{Siteur2014a}. If this assumption is applied then our results on thresholds such as $A_{\max}$ present an approximation to the expectations of the respective quantities when any higher order moments (variance, etc.) of the random variables associated with the description of precipitation are neglected \cite{Siteur2014a}.

While the model extensions (and possibly combinations thereof) presented in Section \ref{sec: Impulsive: Simulations} provide a more realistic description of the ecosystem dynamics under a pulse type precipitation regime, the analytical study of the simpler model \eqref{eq: Impulsive: intro: model discrete dispersal rain water uptake} in Section \ref{sec: Impulsive} is an important tool to gain a better understanding of vegetation patterns in semi-arid environments. Numerical approaches tend to become unreliable as the length of the drought periods increases because decay-type processes of long dry spells reduce the plant density in troughs of the spatial pattern to very small values. This makes numerical integration techniques error-prone and emphasises the importance of analytical pathways into the problem (Figure \ref{fig: Impuslive: Simulations: numerical issue}).

{\color{changes}
The results presented in this paper are based on our analysis of a theoretical model and a comparison with empirical data would be desirable to test these hypotheses. Daily rainfall data is available from the 1980s to the present (see e.g. \cite{Maidment2017} for data from Africa), and data with a coarser temporal scale dates back to the 1940s \cite{Dieulin2019}. However, obtaining high-quality data for vegetation in dryland ecosystems is notoriously difficult due to the large spatial and temporal scales of the ecosystem dynamics. Some limited data obtained from satellite images exists (e.g. \cite{Deblauwe2012}), for example on wavelength which can be used as a proxy for biomass, but a comparison with any model predictions would require a better measure of key ecological properties, as well as a long time series of data points.
}

{\color{changes}In this paper, we have analysed the effects of rainfall intermittency on pattern onset in dryland vegetation in one space dimension only. On flat ground in particular, the consideration of a two-dimensional domain would be a natural extension. This could provide more insight into the patterns' properties such as its type (gap pattern, labyrinth pattern, stripes or spots) under changes to the precipitation regime \cite{Meron2012}. The analysis of the impulsive model on a two-dimensional domain would be significantly more challenging, but methods for studying pattern formation in PDEs on such domains exist (see, for example, \cite{Siero2015} for an analysis of the Klausmeier model), which hold the potential to be adapted to the framework of an impulsive model. }

A further natural area of potential future work could involve an accurate description of overland water flow during a rainfall event. For sloped terrain such a description has been provided and applied to a mathematical model describing the evolution of vegetation patterns by Siteur et al \cite{Siteur2014a}. Their argument is based on water instantaneously flowing downhill and infiltrating the soil in areas of high biomass and can thus not be applied to a flat spatial domain. Indeed, overland flow of water during intense rainfall events on semi-arid flat plains is the subject of ongoing research (e.g. \cite{Thompson2011, Wang2015, Rossi2017}). A detailed description of the overland water flow and infiltration into the soil that occurs before water is consumed by plants relies on a clear distinction between the surface water density and the soil moisture. Such a separation is used in alternative model frameworks \cite{Rietkerk2002, HilleRisLambers2001, Gilad2004}, which could be utilised to include the description of water redistribution during rainfall events under a pulse-type precipitation regime.

The model introduced in this paper is based on the Klausmeier model, which is a model that is deliberately kept simple to facilitate a mathematical analysis of it. A number of more complex models exist (see \cite{Zelnik2013,Borgogno2009, Martinez-Garcia2018} for reviews) that study different aspects of patterned vegetation in more detail by, for example, including two coexisting plant species \cite{Pueyo2010, Baudena2013,Gilad2007a, Eigentler2019Multispecies}, describing water uptake as a nonlocal process \cite{Gilad2007,Gilad2004} or considering effects of nonlocal grazing \cite{Siero2018, Siero2019}. For some of these models numerical studies have investigated the effects of temporal rainfall variability \cite{Guttal2007,Kletter2009, Siteur2014a} and an analytical analysis of those models similar to the work done in this paper could provide further insight how pulse-type phenomena affect patterns in semi-arid environments.

\section*{{\color{changes}Acknowledgements}}
Lukas Eigentler was supported by The Maxwell Institute Graduate School in Analysis and its Applications, a Centre for Doctoral Training funded by the UK Engineering and Physical Sciences Research Council (grant EP/L016508/01), the Scottish Funding Council, Heriot-Watt University and the University of Edinburgh.

\clearpage
\bibliography{bibliography}
\end{document}